\documentclass{amsart}

\usepackage{amsmath,amssymb,amsthm,mathtools}
\usepackage{aliascnt}
\usepackage[alphabetic]{amsrefs}
\usepackage{enumitem}
\usepackage{microtype}
\usepackage[dvipsnames]{xcolor}
\definecolor{RoyalBlue}{HTML}{4169E1}
\definecolor{RoyalRed}{HTML}{9B1B30}
\usepackage{tikz-cd}
\usepackage[colorlinks=true,linkcolor=blue!55!black,citecolor=green!45!black,urlcolor=red!60!black]{hyperref}
\hypersetup{
  pdftitle={Entropic hypocoercivity via the Polchinski flow: an application to lattice sine-Gordon}
}
\numberwithin{equation}{section}

\usepackage{fourier}
\usepackage[margin=1.3in]{geometry}
\usepackage{setspace}

\newtheorem{theorem}{Theorem}[section]
\newaliascnt{proposition}{theorem}
\newtheorem{proposition}[proposition]{Proposition}
\aliascntresetthe{proposition}
\newaliascnt{lemma}{theorem}
\newtheorem{lemma}[lemma]{Lemma}
\aliascntresetthe{lemma}
\newaliascnt{corollary}{theorem}
\newtheorem{corollary}[corollary]{Corollary}
\aliascntresetthe{corollary}
\newaliascnt{assumption}{theorem}

\aliascntresetthe{assumption}
\theoremstyle{definition}
\newaliascnt{definition}{theorem}

\aliascntresetthe{definition}
\newaliascnt{remark}{theorem}
\newtheorem{remark}[remark]{Remark}
\aliascntresetthe{remark}
\usepackage[capitalise,noabbrev]{cleveref}
\crefname{theorem}{Theorem}{Theorems}
\Crefname{theorem}{Theorem}{Theorems}
\crefname{proposition}{Proposition}{Propositions}
\Crefname{proposition}{Proposition}{Propositions}
\crefname{lemma}{Lemma}{Lemmas}
\Crefname{lemma}{Lemma}{Lemmas}
\crefname{corollary}{Corollary}{Corollaries}
\Crefname{corollary}{Corollary}{Corollaries}
\crefname{assumption}{Assumption}{Assumptions}
\Crefname{assumption}{Assumption}{Assumptions}
\crefname{definition}{Definition}{Definitions}
\Crefname{definition}{Definition}{Definitions}
\crefname{remark}{Remark}{Remarks}
\Crefname{remark}{Remark}{Remarks}
\crefname{appendix}{Appendix}{Appendices}
\Crefname{appendix}{Appendix}{Appendices}

\newcommand{\R}{\mathbb R}
\newcommand{\E}{\mathbb E}
\newcommand{\Id}{\mathrm{Id}}
\newcommand{\dd}{\,\mathrm d}
\newcommand{\norm}[1]{\left\lVert #1\right\rVert}
\newcommand{\abs}[1]{\left\lvert #1\right\rvert}
\newcommand{\ip}[2]{\left\langle #1,#2\right\rangle}
\newcommand{\HS}{\mathrm{HS}}
\newcommand{\op}{\mathrm{op}}
\newcommand{\diver}{\operatorname{div}}

\DeclareMathOperator{\Ent}{Ent}
\DeclareMathOperator{\Var}{Var}
\allowdisplaybreaks[1]
\makeatletter
\renewcommand\paragraph{\@startsection{paragraph}{4}{\z@}%
  {1.25ex \@plus .5ex \@minus .2ex}%
  {-.5em}%
  {\normalfont\normalsize\itshape}}
\makeatother

\title[Entropic Hypocoercivity via the Polchinski Flow]{Entropic hypocoercivity via the Polchinski flow:\\
an application to the lattice sine--Gordon model}
\thanks{B.L. is supported by the National Key R\&D Program of China, Grant No.~2024YFA1016000, and the start-up fund 7200843 from City University of Hong Kong. The work of J.L. is supported in part by the National Science Foundation under grant DMS-2309378. J.L. is grateful to Yuansi Chen for many illuminating discussions regarding stochastic localization over the years.}
\author{Bowen Li}
\address{Department of Mathematics, City University of Hong Kong}
\email{boweli4@cityu.edu.hk}

\author{Jianfeng Lu}
\address{Department of Mathematics, Department of Physics, and Department of Chemistry, Duke University}
\email{jianfeng@math.duke.edu}
\date{}

\begin{document}

\begin{abstract}
We develop a multiscale criterion for exponential decay of relative entropy in the canonical stochastic quantization of lattice field theories, the underdamped Langevin dynamics of the field and its momentum. For smooth bounded perturbations of Gaussian fields, the criterion requires the uniform operator-norm bound on the smoothed interaction Hessian along the Polchinski flow of Bauerschmidt and Bodineau to have a small $L^2$ norm in the scale variable. The decay constants depend only on this $L^2$ norm and the friction. The proof measures the transfer of momentum dissipation to position on the same scales that decompose entropy. A phase-space extension of the Polchinski flow gives Gaussian coarse-graining that commutes with free kinetic evolution. Integrating position--momentum cross terms over these scales yields a corrector bounded by entropy. A force defect accounts for changes in the averaged interaction force away from equilibrium; its evolution in scale controls the errors in the kinetic estimate. For the massive two-dimensional lattice sine--Gordon model with $0<\beta<4\pi$ and small dimensionless local coupling, the criterion gives entropy decay for all finite-entropy initial laws, uniformly in the lattice spacing and the volume. For friction proportional to the mass $m$, the rate has optimal order $m$.
\end{abstract}

\maketitle

\begingroup
\small
\setstretch{1}
\tableofcontents
\endgroup

\section{Introduction}
\label{sec:introduction}

Logarithmic Sobolev inequalities that are uniform in the lattice spacing and the volume give quantitative relaxation to equilibrium for the Glauber dynamics of lattice field theories~\cite{BGL}.
Bauerschmidt and Bodineau proved a logarithmic Sobolev inequality for the massive two-dimensional sine--Gordon model below $\beta=6\pi$~\cite[Theorem~1.6]{BB}. Its constant is uniform in the lattice spacing for every coupling and, at small dimensionless local coupling, also in the volume. Their proof uses a multiscale Bakry--\'Emery criterion~\cite[Theorem~1.2]{BB}.
The criterion uses the Polchinski flow, which averages the interaction over the Gaussian fluctuations of the field scale by scale, and assumes at each scale a lower bound, possibly negative, on the Hessian of the resulting renormalized potential, smoothed at that scale; see also~\cite[Theorem~3.6]{BBD2024}.

This paper concerns the kinetic counterpart of the Glauber dynamics, the underdamped Langevin dynamics
\begin{equation}
  \dd q_t=p_t\dd t,\qquad \dd p_t=-\nabla U(q_t)\dd t-\gamma p_t\dd t+\sqrt{2\gamma}\dd W_t,
  \label{eq:intro-langevin}
\end{equation}
with friction $\gamma>0$ and a standard Brownian motion $W$. Its invariant measure is
\[\boldsymbol\mu(\dd q\,\dd p)\propto e^{-U(q)-|p|^2/2}\dd q\,\dd p.\]
For lattice sine--Gordon, $q$ is the lattice field, $U$ is the sine--Gordon Hamiltonian, and \eqref{eq:intro-langevin} is known as the canonical stochastic quantization~\cites{RyangSaitoShigemoto1985,Horowitz1985}. We seek exponential decay of the relative entropy $\Ent(\boldsymbol\varrho_t\mid\boldsymbol\mu)$ of the law $\boldsymbol\varrho_t$ of $(q_t,p_t)$, uniformly in the lattice spacing and the volume.

Two challenges make this difficult. Friction dissipates entropy only through the momentum Fisher information \eqref{eq:ent-physical-dissipation}, which vanishes for densities that depend on position alone. Hypocoercivity recovers decay from the transfer of dissipation by Hamiltonian transport~\cite{Villani}, and its quantitative forms~\cites{DMS,CLW,CGMZ,Lu2026,LiLu2026} assume convexity of $U$ or bounds on its Hessian. In singular field models, these bounds deteriorate as the cutoff is removed.
For lattice sine--Gordon with lattice spacing $\varepsilon$ and coupling $z\neq0$, the microscopic coupling $z_\varepsilon=z\varepsilon^{-\alpha}$ in \eqref{eq:energy}, $\alpha=\beta/(4\pi)$, diverges as $\varepsilon\downarrow0$, and the smallest eigenvalue of the Hessian of the Hamiltonian tends to $-\infty$; see \eqref{eq:intro-microscopic-curvature}.

Our first main result, Theorem~\ref{thm:phase-space-criterion}, is a multiscale criterion for entropy decay that measures the interaction along the Polchinski flow. Let $U(q)=\frac12\ip q{Aq}+V_0(q)$ with $A\succeq\Id$ and $V_0$ smooth and bounded together with all its derivatives. The smoothed interaction Hessian at scale $s$ is the Hessian of the renormalized potential $V_s$ in the coarse variable. If its uniform operator norm has $L^2$ norm less than $1$ as a function of the coarse-graining scale, then every initial law of finite entropy satisfies
\begin{equation}
  \Ent(\boldsymbol\varrho_t\mid\boldsymbol\mu)\le C\,e^{-ct}\,\Ent(\boldsymbol\varrho_0\mid\boldsymbol\mu),\qquad t\ge0,
  \label{eq:intro-decay}
\end{equation}
with explicit constants $c,C>0$ that depend only on $\gamma$ and on an upper bound for this $L^2$ norm of the upper bound of the smoothed Hessian. The estimate is uniform in the dimension and in $A$, and linearization gives the corresponding decay in $L^2(\boldsymbol\mu)$ (Corollary~\ref{cor:abstract-L2}). The microscopic Hessian $\nabla^2V_0$ is the smoothed interaction Hessian at scale zero; a large microscopic Hessian is compatible with the hypothesis when the smoothed interaction Hessian decreases over a short range of small scales.

Its proof, outlined in \S~\ref{sec:proof-ideas}, extends the Polchinski flow to phase space so that coarse-graining commutes with the free kinetic dynamics, and builds a corrector, in the sense of hypocoercivity, from the entropy decomposition of Bauerschmidt and Bodineau~\cite[Section~2.2, (2.16)--(2.17)]{BB}.

Our second main result, Theorem~\ref{thm:bulk-entropy}, applies this criterion to the massive lattice sine--Gordon model with $0<\beta<4\pi$ and small dimensionless local coupling. At unit mass, \eqref{eq:intro-decay} holds for all lattice spacings $0<\varepsilon\le1$ and side lengths $L\ge1$, with constants that depend only on $\gamma$. The hypothesis is verified with the charge-expansion estimates of Bauerschmidt and Bodineau, and their bounds on the Fourier representation of the renormalized potential obtained by the Brydges--Kennedy approach~\cite[Sections~3.3--3.4]{BB}.
The bound \eqref{eq:SG-square-profile} on the smoothed interaction Hessian is of the microscopic size $\varepsilon^{-\alpha}$ only below the lattice scale and decays like $s^{-\alpha/2}$ up to the mass scale.
This power is square integrable at $s=0$ exactly when $\beta<4\pi$, which is the origin of the range of $\beta$.

For friction proportional to the mass $m$, the rate in physical time is of order $m$, and this order is optimal (\S~\ref{subsec:mass-calibration}).
Under an additional small-coupling condition, the static Poincar\'e and logarithmic Sobolev constants of the position law are both of order $m^2$ (Proposition~\ref{prop:static-mass-order}), so this rate is of the order of their square roots. For convex potentials, this square-root acceleration is established in $L^2$ in terms of the Poincar\'e constant~\cite[Theorem~1 and Remark~1.2(ii)]{CLW} and in entropy in terms of the logarithmic Sobolev constant~\cite[Theorem~2.3]{Lu2026}.

Concurrently, Barashkov, Di Ges\`u, and Laarne~\cite{BDL2026} build on the space-time Poincar\'e inequality argument of~\cite{CLW} to obtain exponential $L^2$ convergence for lattice sine--Gordon throughout $\beta<4\pi$; see the discussion after Corollary~\ref{cor:SG-L2}.
They write the interaction as $\cos(\beta_{\mathrm{SPDE}}\varphi)$ and normalize the Gaussian free field as in~\cite{BB}, so their condition $\beta_{\mathrm{SPDE}}^2<4\pi$ is $\beta<4\pi$ in the convention of \eqref{eq:energy}.
For the continuum hyperbolic sine--Gordon dynamics, invariant Gibbs dynamics are constructed in~\cites{ORSW,Zine2025}, and exponential ergodicity on the circle is proved in~\cite{Seong2024}.
The extension to $\beta\ge4\pi$ and the transfer of the uniform estimate to these continuum dynamics remain open.

\subsection*{Use of AI tools.}
The authors explored proof approaches with ChatGPT 5.6 Sol; the proof strategy, which extends the Polchinski flow to phase space and builds a hypocoercive corrector scale by scale, is their own. In preparing the manuscript, the authors used AI assistants, including ChatGPT 5.6 Sol, ChatGPT 6 Astra, and Claude Opus 5.5, to draft and revise the exposition and the presentation of proofs, and to check calculations and the consistency of notation. The authors verified all mathematical statements, proofs, and citations, and take full responsibility for the content of the paper.

\section{Preliminaries}
\label{sec:preliminaries}
In this section, we recall some background material:
\S~\ref{subsec:kinetic-dynamics} defines the underdamped Langevin dynamics. \S~\ref{subsec:polchinski-flow} recalls the Polchinski flow of Bauerschmidt and Bodineau \cite{BB}, which averages the interaction over the Gaussian fluctuations of the position field scale by scale. \S~\ref{subsec:SG-model} defines the lattice sine--Gordon model.

Throughout, $t$ denotes kinetic time and $s$ the coarse-graining scale. Calligraphic letters denote kinetic operators and functionals, sans serif letters denote Polchinski operators, and boldface denotes phase-space measures, block matrices, and observation noises. We write $\E_{\boldsymbol\varrho}\psi=\int\psi\dd\boldsymbol\varrho$. We write $\mathsf N(\cdot,\cdot)$ for  Gaussian measures with the mean as the first argument and the covariance as the second. We write $\nabla^k$ for the tensor of derivatives of order $k$, with $\nabla^0\psi=\psi$. Without a subscript, derivatives act in the position variable $q$; a subscript, as in $\nabla_p$ or $\nabla_{(q,p)}^k$, names the variables. We write $|\cdot|$ for the Euclidean norm on vectors and tensors, $\|\cdot\|_{\mathrm{op}}$ and $\|\cdot\|_{\mathrm{HS}}$ for the operator and Hilbert--Schmidt norms of matrices, and $\operatorname{osc}\psi = \sup\psi - \inf\psi$ for the oscillation of a bounded function $\psi$.

\subsection{Underdamped Langevin dynamics}
\label{subsec:kinetic-dynamics}

Let $A$ be a symmetric $d\times d$ matrix with $A\succeq\Id$.
We consider the class of potentials
\begin{equation} \label{eq:potentialclass}
   U(q)=\frac12\ip q{Aq}+V_0(q)\,, \quad V_0\in C_b^\infty(\R^d),
\end{equation}
where $C_b^\infty$ denotes smooth functions that are bounded together with all their derivatives.
For a fixed finite lattice, the sine--Gordon potential \eqref{eq:energy} belongs to this class, although the bounds need not be uniform in the lattice spacing and the volume; see \eqref{eq:intro-microscopic-curvature}.
The corresponding Gibbs measure on the phase space $\R^d \times \R^d$ is given by, with $Z_0=\E_{\mathsf N(0,A^{-1})}e^{-V_0}$,
\begin{equation}
 \boldsymbol\mu(\dd q\,\dd p) = \nu\otimes\mathsf N(0,\Id) =Z_0^{-1}e^{-V_0(q)}\mathsf N(0,A^{-1})(\dd q)\,\mathsf N(0,\Id)(\dd p)\,,
  \label{eq:ent-abstract-law}
\end{equation}
where the position law $\nu(\dd q) \propto e^{- U(q)} \dd q  \propto e^{-V_0(q)}\mathsf N(0,A^{-1})(\dd q)$ can be viewed as a bounded perturbation of the Gaussian measure $\mathsf N(0,A^{-1})$.

We consider the underdamped Langevin equation, which admits $\boldsymbol\mu$ as its invariant measure,
\begin{equation} \label{eq:langevindy}
  \dd q_t=p_t\dd t,\qquad \dd p_t=-\nabla U(q_t)\dd t-\gamma p_t\dd t+\sqrt{2\gamma}\dd W_t,
\end{equation}
with friction $\gamma>0$, unit temperature, and a standard Brownian motion $W$ in $\R^d$. Its generator, for the semigroup $\mathcal P_t=e^{t\mathcal L}$ acting on observables, can be written as
\begin{equation}
  \mathcal L=\mathcal A+\gamma\mathcal S,\qquad \mathcal A=p\cdot\nabla_q-\nabla U\cdot\nabla_p, \qquad \mathcal S=\Delta_p-p\cdot\nabla_p.
  \label{eq:ent-abstract-generator}
\end{equation}
We write $\mathcal P_t^*=e^{t\mathcal L^*}$ for its adjoint in $L^2(\boldsymbol\mu)$, which acts on densities relative to $\boldsymbol\mu$. More precisely, given an initial law $\boldsymbol\varrho_0=g_0\boldsymbol\mu$ for \eqref{eq:langevindy}, the law of $(q_t,p_t)$ and its relative density satisfy
\begin{equation}
  \boldsymbol\varrho_t=\operatorname{Law}(q_t,p_t)=g_t\boldsymbol\mu,\qquad g_t=\mathcal P_t^*g_0,
  \label{eq:ent-evolving-law}
\end{equation}
that is, $\partial_tg_t=\mathcal L^*g_t=-\mathcal A g_t+\gamma\mathcal S g_t$.
The Markov semigroup $\mathcal P_t^*$ is positivity preserving, fixes constant functions, preserves positive lower and upper bounds, and is a contraction of $L^1(\boldsymbol\mu)$.

We measure the convergence of $\boldsymbol\varrho_t$ to $\boldsymbol\mu$ in relative entropy. For a probability density $f$ with respect to $\boldsymbol\mu$, we write
\begin{equation} \label{eq:fishermoment}
  \Ent_{\boldsymbol\mu}(f)=\int f\log f\dd\boldsymbol\mu,\qquad \mathcal I_p(f)=\int f|\nabla_p\log f|^2\dd\boldsymbol\mu
\end{equation}
for its relative entropy and its momentum Fisher information, with $0\log0=0$. For the probability law $\boldsymbol\varrho=f\boldsymbol\mu$, we also write $\Ent(\boldsymbol\varrho\mid\boldsymbol\mu)=\Ent_{\boldsymbol\mu}(f)$, and we set $\Ent(\boldsymbol\varrho\mid\boldsymbol\mu)=+\infty$ if $\boldsymbol\varrho$ is not absolutely continuous with respect to $\boldsymbol\mu$; the same convention applies to other reference probability measures. Along the dynamics, the entropy is dissipated by the friction part alone: for a smooth initial density with bounded derivatives and a positive lower bound, the skew-adjointness of $\mathcal A$ and Gaussian integration by parts in momentum give
\begin{equation}
  \frac\dd{\dd t}\Ent_{\boldsymbol\mu}(g_t)
  =\int\log g_t\,\mathcal L^*g_t\dd\boldsymbol\mu
  =-\gamma\mathcal I_p(g_t).
  \label{eq:ent-physical-dissipation}
\end{equation}
The differentiation is justified in Lemma~\ref{lem:ent-kinetic-differentiation}. We see from \eqref{eq:ent-physical-dissipation} that the dissipation vanishes for every density that depends only on position.

\subsection{The Polchinski flow} \label{subsec:polchinski-flow}

This subsection recalls the Polchinski flow of Bauerschmidt and Bodineau~\cites{BB,BBD2024}, which decomposes the position law $\nu$ along a scale parameter $s\geq0$. At scale $s$, the flow averages out the fluctuations of $q\sim\nu$ on length scales below $\sqrt s$, that is, on the eigenmodes of $A$ with $s\lambda\gtrsim1$, and produces the renormalized potential $V_s$, which evolves in $s$ by the Polchinski equation. Then, the Hessian of $V_s$, smoothed at scale $s$, can be controlled at every scale, and this smoothed Hessian enters the assumptions of Theorem~\ref{thm:phase-space-criterion}.

We will present the Polchinski flow in the equivalent formalism of stochastic localization~\cites{Eldan2013,ChenEldan2025}, through noising and observation as in~\cite[Sections~1.2--1.3]{Montanari2023}; \cite[Section~4.5]{BBD2024} gives a dictionary between the two points of view. At scale $s$, the position variable $q$ is noised by attenuation and independent Gaussian noise, which produces an observation of $q$ that becomes less informative as $s$ increases. The renormalized measure of~\cite[Section~1.2]{BB} is the law of this observation, up to a linear change of variables, and the renormalized potential gives its density relative to the Gaussian reference measure. The posterior, the conditional law of $q$ given the observation, is the fluctuation measure of~\cite{BBD2024}; it provides the conditional averaging that defines coarse graining in phase space in \S~\ref{subsec:phase-construction}. As $s$ decreases to $0$, the observation determines $q$ and the posterior localizes; stochastic localization follows the posteriors in this direction.

\paragraph{The Gaussian observation.}
Consider a random variable $q \in \R^d$ with distribution $\nu$. At scale $s\geq0$, we define the \emph{observation process} by attenuating it by $Q_s=e^{-sA/2}$, a semigroup associated with $A$, and adding independent Gaussian noise:
\begin{equation} \label{processobv}
  y_s=Q_sq+\xi_s,
\end{equation}
where $\xi_s\sim\mathsf N(0,C_s)$ is independent of $q$.
We denote the law of the observation $y_s$ by $\nu_s$. We choose the noise covariance $C_s$ so that the observation rule preserves the Gaussian reference measure: if $q\sim\mathsf N(0,A^{-1})$, then $y_s\sim\mathsf N(0,A^{-1})$. Indeed, for $q\sim\mathsf N(0,A^{-1})$, the observation $y_s$ is centered Gaussian with covariance $Q_sA^{-1}Q_s+C_s$, which gives $C_s$~\cite[(1.2)]{BB}:
\begin{equation}
  C_s = (\Id-Q_s^2)A^{-1} =\int_0^s e^{-uA}\dd u.
  \label{eq:covdecomp}
\end{equation}
The integral form shows that $C_s$ collects the parts of the scales $u\leq s$ of $A^{-1}=\int_0^\infty e^{-uA}\,\dd u$. As $s$ increases, $C_s$ increases at the rate $\frac\dd{\dd s}C_s=Q_s^2$, from $C_0=0$, when the observation $y_0=q$ is exact, to $C_\infty=A^{-1}$ as $s\to\infty$, when $Q_s\to0$ and the observation no longer depends on $q$. Following \cite{BB}, we call $s\mapsto C_s$ a \emph{covariance decomposition} of $A^{-1}$.

The choice of $C_s$ also makes the observation reversible with respect to $\mathsf N(0,A^{-1})$. If $q\sim\mathsf N(0,A^{-1})$, then $(q,y_s)$ is a centered Gaussian pair whose marginals are both $\mathsf N(0,A^{-1})$ and whose cross-covariance $A^{-1}Q_s$ is symmetric, so $(y_s,q)$ has the same law as $(q,y_s)$. In terms of the observation kernel $\mathsf N(Q_sq,C_s)(\dd y)$ of \eqref{processobv}, this reads
\begin{equation}
  \mathsf N(0,A^{-1})(\dd q)\,\mathsf N(Q_sq,C_s)(\dd y)
  =\mathsf N(Q_sy,C_s)(\dd q)\,\mathsf N(0,A^{-1})(\dd y).
  \label{eq:gaussian-identity}
\end{equation}
By \eqref{eq:gaussian-identity}, when $q\sim\mathsf N(0,A^{-1})$, the conditional law of $q$ given $y_s=y$ is $\mathsf N(Q_sy,C_s)$. Accordingly, for an observation $y\in\R^d$, we split $q=Q_sy+\zeta_q$ into the \emph{coarse field} $Q_sy$, which the observation determines, and the \emph{fluctuation field} $\zeta_q=q-Q_sy$. Moreover, in this case, the fluctuation field has law $\mathsf N(0,C_s)$ and is independent of the observation.
Here, intuitively, the scale $s$ separates the modes of $A$ by their eigenvalues. For $q\sim\mathsf N(0,A^{-1})$, the variance $\lambda^{-1}$ of the component of $q$ along a unit eigenvector of $A$ with eigenvalue $\lambda$ splits between the independent fields $Q_sy_s$ and $\zeta_q$:
\begin{equation*}
 \lambda^{-1}=\underbrace{\lambda^{-1}e^{-s\lambda}}_{\text{coarse field}}+\underbrace{\lambda^{-1}(1-e^{-s\lambda})}_{\text{fluctuation field}}.
\end{equation*}
The observation thus resolves the modes with $s\lambda\ll1$ almost completely, and the fluctuation field carries almost all of the variance of the modes with $s\lambda\gg1$. In the lattice application, large eigenvalues correspond to short wavelengths, so increasing $s$ moves the short wavelengths from the coarse field to the fluctuation field.

\paragraph{The posterior and the renormalized potential.}
We now return to $q\sim\nu$ and denote the joint law of $q$ and its observation $y_s$ by $\pi_s(\dd q\,\dd y) =  \nu(\dd q)\,\mathsf N(Q_sq,C_s)(\dd y)$. Since $\nu(\dd q)=Z_0^{-1}e^{-V_0(q)}\mathsf N(0,A^{-1})(\dd q)$, multiplying the reversibility identity \eqref{eq:gaussian-identity} by $Z_0^{-1}e^{-V_0(q)}$ factors this joint law as
\begin{equation}
  \pi_s(\dd q\,\dd y)=Z_0^{-1}e^{-V_0(q)}\,\mathsf N(Q_sy,C_s)(\dd q)\,\mathsf N(0,A^{-1})(\dd y).
  \label{eq:joint-law}
\end{equation}
By Bayes' formula, \eqref{eq:joint-law} gives the \emph{posterior} of $q$ given $y_s=y$
\begin{equation*}
  \pi_s(\dd q\mid y)\propto e^{-V_0(q)}\,\mathsf N(Q_sy,C_s)(\dd q).
\end{equation*}
To compute its normalizing constant, we define the \emph{renormalized potential} $V_s\in C_b^\infty(\R^d)$ by
\begin{equation}  \label{eq:polchinski2}
   e^{-V_s(q)}=\E_{\zeta_q\sim\mathsf N(0,C_s)}e^{-V_0(q+\zeta_q)},\qquad q\in\R^d.
\end{equation}
Since $q=Q_sy+\zeta_q$, evaluating \eqref{eq:polchinski2} at $Q_sy$ gives the normalizing constant:
\begin{equation} \label{eq:normalconstant}
  \int e^{-V_0(q)}\,\mathsf N(Q_sy,C_s)(\dd q) = e^{-V_s(Q_sy)}.
\end{equation}
We hence have
\begin{equation*}
  \pi_s(\dd q\mid y)=e^{V_s(Q_sy)-V_0(q)}\,\mathsf N(Q_sy,C_s)(\dd q),
\end{equation*}
and, for every bounded measurable function $\psi$,
\begin{equation}
  \E_{\pi_s}[\psi(q)\mid y]
  =e^{V_s(Q_sy)}\E_{\zeta_q\sim\mathsf N(0,C_s)}\!\Big[e^{-V_0(Q_sy+\zeta_q)}\psi(Q_sy+\zeta_q)\Big].
  \label{eq:posterior-average}
\end{equation}
We extend this conditional expectation to phase space in \S~\ref{subsec:phase-construction}.

\paragraph{The observation law and the renormalized measure.}
Integrating \eqref{eq:joint-law} over $q$ and using the normalizing constant \eqref{eq:normalconstant} of the posterior gives the law of the observation
\begin{equation*}
  \nu_s(\dd y)=Z_0^{-1}e^{-V_s(Q_sy)}\,\mathsf N(0,A^{-1})(\dd y).
\end{equation*}
For $s<\infty$, the map $y\mapsto Q_sy$ pushes $\mathsf N(0,A^{-1})$ forward to $\mathsf N(0,A^{-1}-C_s)$, so the coarse field $Q_sy_s$ has law
\begin{equation*}
 (Q_s)_\#\nu_s(\dd q)=Z_0^{-1}e^{-V_s(q)}\,\mathsf N(0,A^{-1}-C_s)(\dd q),
\end{equation*}
which is the renormalized measure of \cite[(1.5)]{BB} and \cite[Definition~3.2, (3.13)]{BBD2024} at scale $s$. We work in the coarse variable $y$, in which the Gaussian reference measure is $\mathsf N(0,A^{-1})$ at every scale $s\in[0,\infty]$, and call $\nu_s$ itself the \emph{renormalized measure}.

Since $\nu$ is the $q$-marginal of $\pi_s$, disintegrating $\pi_s$ along $y$ gives
\begin{equation*}
 \nu(\dd q)=\int\pi_s(\dd q\mid y)\,\nu_s(\dd y),
\end{equation*}
the decomposition of $\nu$ in \cite[(3.17)--(3.19)]{BBD2024}, written in the variable $y$. The posterior $\pi_s(\dd q\mid y)$ is the fluctuation measure at $Q_sy$ \cite{BBD2024}, and its pushforward under $q\mapsto q-Q_sy$ is the conditional law of the fluctuation field $\zeta_q$ given $y_s=y$.

\paragraph{The Polchinski flow.}
We call the family $s\mapsto(V_s,\nu_s)$ the \emph{Polchinski flow} of $V_0$ associated with the covariance decomposition \eqref{eq:covdecomp}. It interpolates between the position law and the Gaussian reference measure:
\begin{equation*}
  \nu_0=\nu,\qquad \nu_s\to\mathsf N(0,A^{-1})\ \text{in total variation as }s\to\infty.
\end{equation*}
The first identity holds since $Q_0=\Id$ and $C_0=0$, and the limit since $Q_s\to0$ and $C_s\to A^{-1}$ in \eqref{eq:polchinski2}.

The scale evolution of $V_s$ follows from the Gaussian convolution \eqref{eq:polchinski2}. For a matrix $M$, write $M:\nabla^2=\sum_{i,j}M_{ij}\partial_{q_i}\partial_{q_j}$. Differentiating \eqref{eq:polchinski2} in $s$ at fixed $q$ and using $\frac\dd{\dd s}C_s=Q_s^2$ gives the heat equation
\begin{equation*}
  \partial_se^{-V_s}=\frac12Q_s^2:\nabla^2e^{-V_s}.
\end{equation*}
By the chain rule,
\begin{equation*}
  Q_s^2:\nabla^2e^{-V_s}=e^{-V_s}\bigl(\abs{Q_s\nabla V_s}^2-Q_s^2:\nabla^2V_s\bigr),
\end{equation*}
so dividing the heat equation by $-e^{-V_s}$ yields the Polchinski equation~\cites{BB,Polchinski1984}
\begin{equation}
  \partial_sV_s=\frac12Q_s^2:\nabla^2V_s-\frac12\abs{Q_s\nabla V_s}^2,\qquad V_s|_{s=0}=V_0.
  \label{eq:polchinski-equation}
\end{equation}
Its quadratic nonlinearity comes only from taking the logarithm of a solution of the linear heat equation.

\subsection{The lattice sine--Gordon model}
\label{subsec:SG-model}

This subsection defines the massive two-dimensional lattice sine--Gordon model, to which our second main result, Theorem~\ref{thm:bulk-entropy}, applies. We use the physical units of~\cite{BB}, in which the inner product, the Laplacian, and the noise are discretizations of their continuum counterparts. For ease of exposition, in \S~\ref{subsec:SG-main} and \S~\ref{subsec:scaling} we pass to the unit mass when stating Theorem~\ref{thm:bulk-entropy} and its proof.

Following \cite[Section~1.3]{BB}, let $\Omega_{\varepsilon,L}=\varepsilon\mathbb Z^2/(L\mathbb Z^2)$ be the discrete torus with lattice spacing $\varepsilon>0$ and side length $L>0$, where $L/\varepsilon\in\mathbb N$. It approximates the continuum torus $L\mathbb T^2$. A lattice field is $\varphi\in\R^{\Omega_{\varepsilon,L}}$, with $(L/\varepsilon)^2$ coordinates. We equip $\R^{\Omega_{\varepsilon,L}}$ with the discrete $L^2$ inner product
\begin{equation*}
  \ip\varphi{\varphi'}_\varepsilon=\varepsilon^2\sum_{x\in\Omega_{\varepsilon,L}}\varphi_x\varphi'_x, \qquad \norm\varphi_\varepsilon^2=\ip\varphi\varphi_\varepsilon,
\end{equation*}
which approximates $\int_{L\mathbb T^2}\varphi\varphi'\dd x$ as $\varepsilon\to0$. Gradients, adjoints, Gaussian measures, and Brownian motions on $\R^{\Omega_{\varepsilon,L}}$ are understood with respect to $\ip\cdot\cdot_\varepsilon$.
In particular, we write $\nabla_\varepsilon$ for the gradient in this inner product, so that $(\nabla_\varepsilon\psi)_x=\varepsilon^{-2}\partial_{\varphi_x}\psi$ for a function $\psi$ on $\R^{\Omega_{\varepsilon,L}}$, and the nearest-neighbor lattice Laplacian is given by
\begin{equation*}
  (\Delta_\varepsilon\varphi)_x := \varepsilon^{-2}\sum_{j=1}^2 \bigl(\varphi_{x+\varepsilon e_j}+\varphi_{x-\varepsilon e_j}-2\varphi_x\bigr),
\end{equation*}
where $e_1,e_2$ are the coordinate unit vectors. For mass $m>0$, sine--Gordon parameter $\beta>0$, and coupling parameter $z\in\R$, we set
\begin{equation}
  A_\varepsilon := -\Delta_\varepsilon+m^2\Id, \qquad \alpha := \frac{\beta}{4\pi}, \qquad z_\varepsilon := z\varepsilon^{-\alpha}.
  \label{eq:parameters}
\end{equation}
The operator $A_\varepsilon$ is symmetric with respect to $\ip\cdot\cdot_\varepsilon$ and satisfies $A_\varepsilon\succeq m^2\Id$. The Hamiltonian, in the convention of Bauerschmidt and Bodineau~\cite{BB}, is
\begin{equation}
  H_{\varepsilon,L}(\varphi) = \frac12\ip{\varphi}{A_\varepsilon\varphi}_\varepsilon +2z_\varepsilon\varepsilon^2 \sum_{x\in\Omega_{\varepsilon,L}}\cos\big(\sqrt\beta\,\varphi_x \big).
  \label{eq:energy}
\end{equation}
The underdamped Langevin dynamics with friction $\gamma>0$ is
\begin{equation}
  \dd\varphi_t=\varpi_t\dd t,\qquad \dd\varpi_t=-\nabla_\varepsilon H_{\varepsilon,L}(\varphi_t)\dd t -\gamma\varpi_t\dd t+\sqrt{2\gamma}\dd W_t^\varepsilon,
  \label{eq:langevin}
\end{equation}
where $\varpi$ is the momentum field conjugate to $\varphi$, and the Brownian motion $W^\varepsilon$ on $(\R^{\Omega_{\varepsilon,L}},\ip\cdot\cdot_\varepsilon)$ has coordinates $W^\varepsilon_x=\varepsilon^{-1}W_x$ for independent standard one-dimensional Brownian motions $(W_x)_{x\in\Omega_{\varepsilon,L}}$. At each fixed lattice, the drift is globally Lipschitz and the process is globally well posed. Equation~\eqref{eq:langevin} is the \emph{canonical stochastic quantization}~\cites{RyangSaitoShigemoto1985,Horowitz1985} of the lattice sine--Gordon model.
Its invariant measure is
\begin{equation}
  \boldsymbol\mu_{\varepsilon,L}=\nu_{\varepsilon,L}\otimes\kappa_{\varepsilon,L}, \quad \nu_{\varepsilon,L}(\dd\varphi)\propto e^{-H_{\varepsilon,L}(\varphi)}\dd\varphi, \quad \kappa_{\varepsilon,L}(\dd\varpi)\propto e^{-\norm\varpi_\varepsilon^2/2}\dd\varpi.
  \label{eq:invariant-law}
\end{equation}
The position marginal $\nu_{\varepsilon,L}$ is the \emph{lattice sine--Gordon measure}. It perturbs the massive lattice Gaussian free field $\mathsf N(0,A_\varepsilon^{-1})$:
\begin{equation*}
  \nu_{\varepsilon,L}(\dd\varphi)\propto\exp\Bigl(-2z_\varepsilon\varepsilon^2\sum_{x\in\Omega_{\varepsilon,L}}\cos\big(\sqrt\beta\,\varphi_x\big)\Bigr)\mathsf N(0,A_\varepsilon^{-1})(\dd\varphi).
\end{equation*}
The factor $\varepsilon^{-\alpha}$ in $z_\varepsilon$ is the ultraviolet normalization of \cite[(1.16)]{BB}. For fixed $m,L>0$, the site variance of the massive lattice Gaussian free field satisfies
\begin{equation*}
  \Var_{\mathsf N(0,A_\varepsilon^{-1})}(\varphi_x)
  =\frac1{2\pi}\log\varepsilon^{-1}+O_{m,L}(1),\qquad \varepsilon\downarrow0,
\end{equation*}
so the cosine has Gaussian mean
\begin{equation*}
  \E_{\mathsf N(0,A_\varepsilon^{-1})}\cos(\sqrt\beta\,\varphi_x)
  =\exp\!\left(-\frac{\beta}{2}\Var_{\mathsf N(0,A_\varepsilon^{-1})}(\varphi_x)\right)
  =\varepsilon^{\alpha}e^{O_{m,L}(1)}\longrightarrow0 .
\end{equation*}
Exact Wick ordering with respect to $\mathsf N(0,A_\varepsilon^{-1})$ is then defined by
\begin{equation*}
  :\!\cos(\sqrt\beta\,\varphi_x)\!:\;
  =\exp\!\left(\frac{\beta}{2}\Var_{\mathsf N(0,A_\varepsilon^{-1})}(\varphi_x)\right)
     \cos(\sqrt\beta\,\varphi_x),
\end{equation*}
so that its Gaussian mean is one.
The coupling $z_\varepsilon=z\varepsilon^{-\alpha}$ keeps only the leading factor $\varepsilon^{-\alpha}$ of the exact Wick factor. Equivalently, the interaction in \eqref{eq:energy} is a Wick-ordered cosine interaction whose coupling is $z$ times a factor $e^{O_{m,L}(1)}$, which is bounded above and below by positive constants as $\varepsilon\downarrow0$ and depends on $m$, $L$, and the lattice regularization. The exponent $\alpha$ is the scaling dimension of $:\!e^{i\sqrt\beta\varphi}\!:$. With this normalization, $\nu_{\varepsilon,L}$ converges to the continuum massive sine--Gordon measure for $0<\beta<8\pi$ and small $z$; see \cite[Section~1.3]{BB} for precise statements and references.

We denote the semigroup of \eqref{eq:langevin} acting on observables by $\mathcal P_t^{\varepsilon,L}$ and its adjoint in $L^2(\boldsymbol\mu_{\varepsilon,L})$ by $(\mathcal P_t^{\varepsilon,L})^*$. The adjoint evolves densities relative to the invariant measure: $g_t=(\mathcal P_t^{\varepsilon,L})^*g_0$, and $\boldsymbol\varrho_t=g_t\boldsymbol\mu_{\varepsilon,L}$ is the law of $(\varphi_t,\varpi_t)$. At each fixed lattice we suppress the parameters $\varepsilon,L$ on these densities and laws. For a probability density $f$ relative to $\boldsymbol\mu_{\varepsilon,L}$, the entropy convention of \S~\ref{subsec:kinetic-dynamics} reads
\begin{equation}
  \Ent(f\boldsymbol\mu_{\varepsilon,L}\mid\boldsymbol\mu_{\varepsilon,L})=\Ent_{\boldsymbol\mu_{\varepsilon,L}}(f)=\int f\log f\dd\boldsymbol\mu_{\varepsilon,L},\qquad \E_{\boldsymbol\mu_{\varepsilon,L}}f=1.
  \label{eq:entropy-definition}
\end{equation}

\section{Main results}

\label{sec:criterion}
Our main results give exponential convergence to equilibrium in relative entropy. The first is a criterion for underdamped Langevin dynamics, formulated in terms of the interaction along the Polchinski flow (\S~\ref{subsec:entropy-criterion}). The second applies this criterion to lattice sine--Gordon and gives a decay rate uniform in the lattice spacing and the volume (\S~\ref{subsec:SG-main}). We then explain the mechanism of the proof in \S~\ref{sec:proof-ideas}.

\subsection{Entropy decay under a smoothed Hessian bound}
\label{subsec:entropy-criterion}
In the coarse variable $y$, the renormalized potential reads $V_s(Q_sy)$, and we call its Hessian the \emph{smoothed interaction Hessian}. By the chain rule, we have
\begin{equation*}
  \nabla_y^2[V_s(Q_sy)]=Q_s\nabla^2V_s(Q_sy)Q_s.
\end{equation*}
We denote its uniform operator-norm bound by
\begin{equation*}
  b(s)=\sup_{y\in\R^d}\norm{Q_s\nabla^2V_s(Q_sy)Q_s}_\op=\sup_{q\in\R^d}\norm{Q_s\nabla^2V_s(q)Q_s}_\op,
\end{equation*}
where the two suprema agree because $Q_s$ is invertible. We introduce the parameter
\begin{equation}
 \eta=\left(\int_0^\infty b(s)^2\dd s\right)^{1/2}.
  \label{eq:ent-abstract-budget}
\end{equation}
At $s=0$, $b(0)=\sup_q\norm{\nabla^2V_0(q)}_\op$ is the microscopic curvature bound. For $s>0$, each factor $Q_s$ multiplies an eigenmode of $A$ with eigenvalue $\lambda$ by $e^{-s\lambda/2}$. The modes with $s\lambda\gg1$, whose fluctuations the Gaussian convolution \eqref{eq:polchinski2} has already averaged into $V_s$, enter $b(s)$ only with exponentially small weight. Thus $\eta$ measures the curvature of the averaged interaction through the $L^2$ norm of $b$ in the scale variable.

The first main result of the paper is the following entropy decay estimate.
\begin{theorem}[Entropy decay under a smoothed Hessian bound]
\label{thm:phase-space-criterion}
Let $A\succeq\Id$, $V_0\in C_b^\infty(\R^d)$, and $\gamma>0$ be as in \S~\ref{subsec:kinetic-dynamics}, and assume $\eta\leq\eta_*<1$. For every initial law $\boldsymbol\varrho_0$ with $\Ent(\boldsymbol\varrho_0\mid\boldsymbol\mu)<\infty$, the law $\boldsymbol\varrho_t$ at time $t$ of \eqref{eq:langevindy} with initial law $\boldsymbol\varrho_0$ satisfies
\begin{equation}
  \Ent(\boldsymbol\varrho_t\mid\boldsymbol\mu) \leq C_{\gamma,\eta_*}\,e^{-c_{\gamma,\eta_*}t}\,\Ent(\boldsymbol\varrho_0\mid\boldsymbol\mu), \qquad t\geq0,
  \label{eq:ent-abstract-decay}
\end{equation}
with constants $c_{\gamma,\eta_*},C_{\gamma,\eta_*}>0$ that depend only on the friction $\gamma$ and the fixed bound $\eta_*$.
\end{theorem}

The proof gives \eqref{eq:ent-abstract-decay} with the explicit constants
\begin{equation}
  c_{\gamma,\eta_*}=\frac{(1-\eta_*)\theta_{\gamma,\eta_*}}{1+\theta_{\gamma,\eta_*}}, \qquad C_{\gamma,\eta_*}=\frac{1+\theta_{\gamma,\eta_*}}{1-\theta_{\gamma,\eta_*}},
  \label{eq:ent-general-constants}
\end{equation}
where
\begin{equation}
  \theta_{\gamma,\eta_*}=\min\left\{\frac{\gamma}{2K_{\gamma,\eta_*}}, \frac{\gamma(1-\eta_*)}{2(\gamma+\eta_*)^2}\right\}, \qquad K_{\gamma,\eta_*}=\max\{2,\gamma\}+\frac{\eta_*}{\sqrt2}.
  \label{eq:ent-general-theta}
\end{equation}
Since $0<\theta_{\gamma,\eta_*}<\frac12$, we have $c_{\gamma,\eta_*}>0$ and $1<C_{\gamma,\eta_*}<3$, and the rate $c_{\gamma,\eta_*}$ tends to zero as $\eta_*\uparrow1$. Because the constants depend only on $\gamma$ and $\eta_*$, the estimate is uniform in the dimension $d$, in the matrix $A\succeq\Id$, and in the interaction $V_0$ subject to $\eta\le\eta_*$; in particular, the oscillation of $V_0$ and the size of its gradient enter only through $\eta$.

\begin{remark}
  We compare Theorem~\ref{thm:phase-space-criterion} with the multiscale Bakry--\'Emery criterion of Bauerschmidt and Bodineau~\cite[Theorem~1.2]{BB}. That criterion assumes, at each scale, a lower bound on the smoothed interaction Hessian by a real constant that may be negative, and gives a logarithmic Sobolev inequality for $\nu$. Theorem~\ref{thm:phase-space-criterion} assumes, at each scale, the two-sided bound $\norm{Q_s\nabla^2V_s(Q_sy)Q_s}_\op\le b(s)$, with $\eta<1$, and gives entropy decay for the underdamped Langevin dynamics.
\end{remark}

\smallskip

Near equilibrium, relative entropy is quadratic to leading order: $\Ent_{\boldsymbol\mu}(1+\kappa\psi)\approx\frac{\kappa^2}{2}\norm\psi_2^2$ for small $\kappa$ and centered $\psi$. Linearizing the entropy estimate at the equilibrium density $f=1$ therefore gives the $L^2$ decay for $\mathcal P_t^*$, and duality gives the same bound for $\mathcal P_t$.

\begin{corollary}[$L^2$ decay from entropy decay]
\label{cor:abstract-L2}
Under the assumptions of Theorem~\ref{thm:phase-space-criterion}, every $\psi\in L^2(\boldsymbol\mu)$ with $\E_{\boldsymbol\mu}\psi=0$ satisfies
\begin{equation*}
  \norm{\mathcal P_t\psi}_2\le\sqrt{C_{\gamma,\eta_*}}\,e^{-c_{\gamma,\eta_*}t/2}\norm\psi_2.
\end{equation*}
The same estimate holds for $\mathcal P_t^*$.
\end{corollary}

\begin{remark}
  Concurrently, Barashkov, Di Ges\`u, and Laarne~\cite{BDL2026} prove exponential $L^2$ convergence for lattice sine--Gordon throughout $\beta<4\pi$, the range of Corollary~\ref{cor:SG-L2}, with estimates uniform in the lattice spacing for every coupling. Their analysis builds on the argument of Cao, Lu, and Wang~\cite{CLW}. It assumes a logarithmic Sobolev inequality for the position law, which for sine--Gordon is that of~\cite{BB}, and exponential lower-tail bounds, under the equilibrium measure, for the bottom of the spectrum of the lattice Hessian $\nabla_\varepsilon^2H_{\varepsilon,L}$, a random Schr\"odinger operator; these bounds are localized, so that the volume enters only through the logarithmic Sobolev constant. For friction proportional to the mass, Theorem~\ref{thm:bulk-entropy} gives entropy decay at a rate of order $m$, uniformly in the lattice spacing and the volume.
\end{remark}

\subsection{Uniform entropy decay for lattice sine--Gordon}
\label{subsec:SG-main}
For lattice sine--Gordon, the criterion $\eta < 1$ gives the same entropy decay constants as the lattice is refined and the volume is enlarged.
The change of variables of \S~\ref{subsec:scaling}, which measures lengths and times in units of $m^{-1}$, maps \eqref{eq:langevin} to the same dynamics at unit mass, with the dimensionless parameters $m\varepsilon$, $mL$, $z\,m^{-(2-\alpha)}$, and $\gamma/m$ in place of $(\varepsilon,L,z,\gamma)$. With a slight abuse of notation, we retain the notations $\varepsilon$, $L$, $z$, and $\gamma$ for the rescaled parameters and state the results at $m=1$. The isometry $q=\varepsilon\varphi$ of \S~\ref{subsec:scaling} then writes $H_{\varepsilon,L}$ as a potential $U$ of the form \eqref{eq:potentialclass}, with $A=A_\varepsilon=-\Delta_\varepsilon+\Id\succeq\Id$ and with $V_0$ given by the cosine term of \eqref{eq:energy} in the variable $q$, and it preserves Hessian operator norms. We write $b_{\varepsilon,L}$ and $\eta_{\varepsilon,L}$ for the resulting function $b$ and parameter $\eta$ of \S~\ref{subsec:entropy-criterion}.
The second main result of the paper is the following uniform entropy decay estimate.

\begin{theorem}[Uniform entropy decay for lattice sine--Gordon]
\label{thm:bulk-entropy}
Let $0<\beta<4\pi$ and $m=1$. There exists $z_*(\beta)>0$ such that, for $\gamma>0$ and
\begin{equation}
  0<\varepsilon\le1,\qquad L\ge1,\qquad L/\varepsilon\in\mathbb N,\qquad |z|\le z_*(\beta),
  \label{eq:bulk-entropy-regime}
\end{equation}
the law $\boldsymbol\varrho_t$ of \eqref{eq:langevin}, started from any $\boldsymbol\varrho_0$ with $\Ent(\boldsymbol\varrho_0\mid\boldsymbol\mu_{\varepsilon,L})<\infty$, satisfies
\begin{equation}
  \Ent(\boldsymbol\varrho_t\mid\boldsymbol\mu_{\varepsilon,L}) \le C_\gamma e^{-c_\gamma t} \Ent(\boldsymbol\varrho_0\mid\boldsymbol\mu_{\varepsilon,L}),\qquad t\ge0,
  \label{eq:bulk-entropy-decay}
\end{equation}
with constants $c_\gamma,C_\gamma>0$ that depend only on $\gamma$.
\end{theorem}

We will see in Proposition~\ref{prop:SG-square-budget} that $\eta_{\varepsilon,L}\le1/2$ holds in the regime \eqref{eq:bulk-entropy-regime}, and hence the proof applies Theorem~\ref{thm:phase-space-criterion} with $\eta_*=1/2$. We write $c_\gamma$, $C_\gamma$, $\theta_\gamma$, $K_\gamma$ for the constants \eqref{eq:ent-general-constants} and \eqref{eq:ent-general-theta} at this value:
\begin{equation}
  c_\gamma=\frac{\theta_\gamma}{2(1+\theta_\gamma)}, \qquad C_\gamma=\frac{1+\theta_\gamma}{1-\theta_\gamma},
  \label{eq:bulk-entropy-constants}
\end{equation}
where
\begin{equation*}
  \theta_\gamma=\min\left\{\frac{\gamma}{2K_\gamma},\frac{\gamma}{4(\gamma+1/2)^2}\right\}, \qquad K_\gamma=\max\{2,\gamma\}+\frac1{2\sqrt2}.
\end{equation*}
As in \S~\ref{subsec:entropy-criterion}, $1<C_\gamma<3$, and linearization gives the corresponding $L^2$ estimate (Corollary~\ref{cor:SG-L2}).

The smallness assumption of Theorem~\ref{thm:bulk-entropy} concerns the dimensionless local coupling $z$. To see how the interaction behaves as $\varepsilon\downarrow0$ and $L\to\infty$, consider the constant fields $\varphi_x=c$, $x\in\Omega_{\varepsilon,L}$.
The cosine term $2z_\varepsilon\varepsilon^2\sum_x\cos(\sqrt\beta\,\varphi_x)$ in \eqref{eq:energy} is a sum over the $(L/\varepsilon)^2$ sites of $\Omega_{\varepsilon,L}$, and it attains its extreme values $\pm2|z|L^2\varepsilon^{-\alpha}$ at the constant fields with $z\cos(\sqrt\beta\,c)=\pm|z|$.
At a maximizing field, for example $c=0$ if $z\ge0$ and $c=\pi/\sqrt\beta$ if $z<0$, the Hessian $\nabla_\varepsilon^2H_{\varepsilon,L}$ with respect to $\ip\cdot\cdot_\varepsilon$ and its smallest eigenvalue, whose eigenvectors are the constant fields, are
\begin{equation}
 \nabla_\varepsilon^2H_{\varepsilon,L}(\varphi)=A_\varepsilon-2\beta|z|\varepsilon^{-\alpha}\Id,\qquad \lambda_{\min}\bigl(\nabla_\varepsilon^2H_{\varepsilon,L}(\varphi)\bigr)=1-2\beta|z|\varepsilon^{-\alpha}.
  \label{eq:intro-microscopic-curvature}
\end{equation}
For any fixed nonzero $z$ in the regime of the theorem, the extreme values of the cosine term are unbounded as $\varepsilon\downarrow0$ or $L\to\infty$, and the smallest eigenvalue in \eqref{eq:intro-microscopic-curvature} tends to $-\infty$ as $\varepsilon\downarrow0$. In particular, the potential is nonconvex for small $\varepsilon$.

Theorem~\ref{thm:bulk-entropy} is uniform in these limits because the constants of Theorem~\ref{thm:phase-space-criterion} depend on the interaction only through $\eta$, and the scale dependence of $b_{\varepsilon,L}$ keeps $\eta_{\varepsilon,L}$ bounded. The charge-expansion estimates of Bauerschmidt and Bodineau control $b_{\varepsilon,L}$ at every scale: in the regime of Theorem~\ref{thm:bulk-entropy}, the proof of Proposition~\ref{prop:SG-square-budget} gives
\begin{equation}
  b_{\varepsilon,L}(s)\leq C_\beta|z|
  \begin{cases}
    \varepsilon^{-\alpha},&0\leq s\leq\varepsilon^2,\\
    s^{-\alpha/2},&\varepsilon^2\leq s\leq1,\\
    e^{-s},&s\geq1,
  \end{cases}
  \qquad \alpha=\frac{\beta}{4\pi},
  \label{eq:SG-square-profile}
\end{equation}
with $C_\beta$ independent of $\varepsilon$ and $L$. Below the lattice scale, the bound has the order of the microscopic Hessian. Above it, the bound is determined by the observation scale and is uniform in the cutoff: the power $s^{-\alpha/2}$ describes the strength of the Wick-ordered interaction after fluctuations up to length $\sqrt{s}$ have been averaged out, and the mass gives exponential decay beyond scale one. The restriction on $\beta$ comes from the power law, through the square integrability condition
\begin{equation*}
  \int_0^1s^{-\alpha}\dd s<\infty
  \quad\Longleftrightarrow\quad \alpha<1
  \quad\Longleftrightarrow\quad \beta<4\pi.
\end{equation*}
Under this condition, the contribution of the interval $[0,\varepsilon^2]$ to $\eta_{\varepsilon,L}^2$, of order $z^2\varepsilon^{2-2\alpha}$, vanishes as $\varepsilon\downarrow0$.

For general mass $m>0$, undoing the scaling of \S~\ref{subsec:scaling} gives the physical-time rate $c_\gamma m$, with $\varepsilon$, $L$, $z$, and $\gamma$ in the hypotheses and constants understood as the rescaled parameters. For friction proportional to the mass, the Gaussian constant mode shows that this order in $m$ is optimal for the stated family; see \S~\ref{subsec:mass-calibration}.

\subsection{Strategy of the proof}
\label{sec:proof-ideas}
We outline the ideas behind Theorem~\ref{thm:phase-space-criterion}. Hamiltonian transport turns position dependence into momentum dependence, and hypocoercivity quantifies this transfer by adding to the entropy a cross term that couples position and momentum derivatives~\cite{Villani}. The kinetic derivative of a cross term built from the density itself contains the Hessian of $U$, whose smallest eigenvalue for lattice sine--Gordon tends to $-\infty$ as $\varepsilon \to 0$; see \eqref{eq:intro-microscopic-curvature}. We build the cross term from the coarse-grained densities of the Polchinski flow, in which the interaction enters through the smoothed interaction Hessian, and integrate it over all scales; this is the corrector \eqref{eq:ent-corrector} of \S~\ref{subsec:corrector-size}. Sections~\ref{sec:scale-calculus}--\ref{sec:modified-entropy} carry out the argument for smooth initial densities bounded away from zero, and Appendix~\ref{app:regularity} supplies the regularity estimates and the approximation of general initial laws of finite entropy.

Differentiating the cross term at a fixed scale requires a kinetic equation for the coarse-grained density: evolving a law and then observing it should give the same law as observing first and then evolving (Figure~\ref{fig:kinetic-observation}). The position observation is that of the Polchinski flow, and we match the momentum observation to it (\S~\ref{subsec:phase-construction}). On an eigenmode of $A$, with position and momentum measured in their equilibrium units, harmonic motion is a rotation; attenuating both coordinates by the same factor and adding isotropic Gaussian noise commutes with this rotation and with the momentum friction and noise. For the interacting dynamics, the diagram commutes up to one forcing term, the force defect, the difference between the interaction force averaged over the unresolved fluctuations under the evolving law and the same average at equilibrium, which is the renormalized force (Lemma~\ref{lem:ent-intertwining}).

The corrector is built on the entropy decomposition of Bauerschmidt and Bodineau, which writes entropy as the information that the observation loses about the state, summed over all scales (Lemma~\ref{lem:ent-decomposition}). The rate of loss is a Fisher information that involves both the position and momentum scores, the logarithmic gradients of the coarse-grained density, whereas friction dissipates entropy through momentum alone. The corrector pairs the two scores at each scale and integrates the pairing over scales. Free Hamiltonian transport rotates position dependence into momentum dependence at every scale, so the kinetic derivative of the corrector produces the position part of the entropy decomposition with a negative sign (Lemma~\ref{lem:ent-exact-cross-evolution}). The decomposition also bounds the corrector by entropy, so the modified entropy is comparable to entropy (Lemma~\ref{lem:ent-corrector-size}).

The remaining terms in the kinetic derivative of the corrector involve the coarse-grained densities at every scale, while friction controls only the momentum Fisher information of the density itself, at scale zero. Coarse-graining bridges this mismatch of scales. In the momentum variable, the Polchinski flow is an Ornstein--Uhlenbeck smoothing determined by $A$, along which the momentum Fisher information decreases, so the control by friction extends to every scale (Lemma~\ref{lem:ent-velocity-fisher}). The interaction enters only through the force defect, which vanishes at scale zero and is generated along the flow by the smoothed interaction Hessian acting on position derivatives. Since the entropy decomposition controls position derivatives in mean square over scales, the accumulated force defect is measured against entropy by the $L^2$ norm $\eta$ of the bound $b$ in the scale variable (Proposition~\ref{prop:ent-force-energy}). It removes at most a fraction $\eta$ of the negative entropy produced by free Hamiltonian transport, and the hypothesis $\eta\leq\eta_*<1$ leaves the rest to drive the decay (Proposition~\ref{prop:ent-corrector-differential}).

\suppressfloats[t]
\section{Coarse graining and kinetic evolution}
\label{sec:scale-calculus}

We construct the observation laws and derive their evolution in both coarse-graining scale and kinetic time, establishing the commuting diagram in Figure~\ref{fig:kinetic-observation}, which is fundamental for constructing the corrector in the modified entropy functional.

\subsection{Matched coarse graining in phase space}
\label{subsec:phase-construction}

In this subsection, we extend the position-space Polchinski flow of \S~\ref{subsec:polchinski-flow} to phase space through a Gaussian observation of position and momentum. We match the momentum observation to the position observation such that the observation process commutes with the free kinetic dynamics, \eqref{eq:langevindy} with $V_0=0$. Applied at each kinetic time, this observation gives the vertical arrows of Figure~\ref{fig:kinetic-observation}.

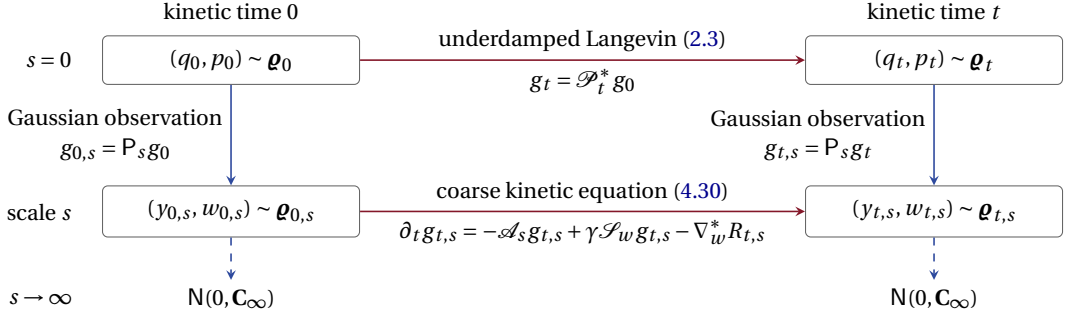
\begin{figure}[t]
\centering
\begin{tikzpicture}[>=stealth,
    state/.style={draw=black!60,rounded corners=2pt,inner sep=5pt,
      minimum width=3.4cm,font=\small},
    lab/.style={font=\small,align=center},
    scale arrow/.style={->,semithick,draw=RoyalBlue!70!black},
    kinetic arrow/.style={->,semithick,draw=RoyalRed!85!black}]
  \node[lab] at (3,-0.6) {kinetic time $0$};
  \node[lab] at (12.3,-0.6) {kinetic time $t$};
  \node[lab,anchor=east] at (1,-1.25) {$s=0$};
  \node[lab,anchor=east] at (1,-3.25) {scale $s$};
  \node[lab,anchor=east] at (1,-4.4) {$s\to\infty$};
  \node[state] (f0) at (3,-1.25) {$(q_0,p_0)\sim\boldsymbol\varrho_0$};
  \node[state] (ft) at (12.3,-1.25) {$(q_t,p_t)\sim\boldsymbol\varrho_t$};
  \node[state] (c0) at (3,-3.25) {$(y_{0,s},w_{0,s})\sim\boldsymbol\varrho_{0,s}$};
  \node[state] (ct) at (12.3,-3.25) {$(y_{t,s},w_{t,s})\sim\boldsymbol\varrho_{t,s}$};
  \node[lab] (e0) at (3,-4.4) {$\mathsf N(0,\mathbf C_\infty)$};
  \node[lab] (et) at (12.3,-4.4) {$\mathsf N(0,\mathbf C_\infty)$};
  \draw[kinetic arrow] (f0) -- node[lab,above]{underdamped Langevin \eqref{eq:langevindy}}
    node[lab,below]{$g_t=\mathcal P_t^*g_0$} (ft);
  \draw[kinetic arrow] (c0) -- node[lab,above]{coarse kinetic equation \eqref{eq:ent-coarse-density-equation}}
    node[lab,below]{$\partial_tg_{t,s}=-\mathcal A_sg_{t,s}+\gamma\mathcal S_wg_{t,s}-\nabla_w^*R_{t,s}$} (ct);
  \draw[scale arrow] (f0) -- node[lab,left]{Gaussian observation\\
    $g_{0,s}=\mathsf P_sg_0$} (c0);
  \draw[scale arrow] (ft) -- node[lab,left]{Gaussian observation\\
    $g_{t,s}=\mathsf P_sg_t$} (ct);
  \draw[scale arrow,dashed] (c0) -- (e0);
  \draw[scale arrow,dashed] (ct) -- (et);
\end{tikzpicture}
\caption{Kinetic evolution and Gaussian observation. Horizontal (red) arrows follow kinetic time at fixed scale, and vertical (blue) arrows follow scale at fixed kinetic time, along the phase-space Polchinski flow. The state $(q_t,p_t)$ has law $\boldsymbol\varrho_t$, and its observation $(y_{t,s},w_{t,s})$ has law $\boldsymbol\varrho_{t,s}$; the densities $g_t$ and $g_{t,s}$ on the arrows are relative to $\boldsymbol\mu$ and $\boldsymbol\mu_s$, and $\mathsf P_s$ is the coarse-graining operator \eqref{eq:ent-coarse-operator}. The lower row is the exact evolution of $g_{t,s}$, with forcing by the force defect $R_{t,s}$; for $V_0=0$, the force defect vanishes and both rows are the free kinetic dynamics. As $s\to\infty$, the observation laws tend to the Gaussian law (dashed arrows).}
\label{fig:kinetic-observation}
\end{figure}

\paragraph{Matched Gaussian decomposition.}
With $\mathbf C_\infty:=\operatorname{diag}(A^{-1},\Id)$, the invariant measure \eqref{eq:ent-abstract-law} is $\boldsymbol\mu(\dd q\,\dd p)=Z_0^{-1}e^{-V_0(q)}\,\mathsf N(0,\mathbf C_\infty)(\dd q\,\dd p)$, a perturbation of the Gaussian reference measure $\mathsf N(0,\mathbf C_\infty)$ by the interaction $V_0$, which depends on position alone. We attenuate position and momentum by the same matrix $Q_s$, and write $\mathbf Q_s=\operatorname{diag}(Q_s,Q_s)$. At scale $s\geq0$, the Gaussian observation of a state $(q,p)$ is
\begin{equation} \label{eq:phasespaceobv}
  \begin{pmatrix}y_s\\w_s\end{pmatrix}
  =\mathbf Q_s\begin{pmatrix}q\\p\end{pmatrix}+\boldsymbol\xi_s,
  \qquad \boldsymbol\xi_s\sim\mathsf N(0,\mathbf C_s),
\end{equation}
with $(q,p)$ and $\boldsymbol\xi_s$ independent. Here $(y_s,w_s)$ is a noisy observation of $(q,p)$ in phase space. As in \eqref{eq:covdecomp}, the observation process \eqref{eq:phasespaceobv} preserves the reference measure $\mathsf N(0,\mathbf C_\infty)$ if and only if $\mathbf C_\infty=\mathbf Q_s\mathbf C_\infty\mathbf Q_s+\mathbf C_s$ for all $s$, which determines the noise covariance. Since $\mathbf Q_s$ commutes with $\mathbf C_\infty$,
\begin{equation}
  \mathbf C_s = (\Id-\mathbf Q_s^2)\mathbf C_\infty = \begin{pmatrix}C_s&0\\0&AC_s\end{pmatrix},\qquad
  \frac\dd{\dd s}\mathbf C_s=\begin{pmatrix}Q_s^2&0\\0&AQ_s^2\end{pmatrix}.
  \label{eq:ent-BB-block-covariance}
\end{equation}
We call this \emph{matched coarse graining}, since the same fraction $\Id-Q_s^2$ of the equilibrium variance is assigned to the fluctuation in position and in momentum.
Since $\mathbf C_0=0$, $\frac\dd{\dd s}\mathbf C_s\succeq0$, and $\mathbf C_s\to\mathbf C_\infty$ as $s\to\infty$, the family $s\mapsto\mathbf C_s$ is a covariance decomposition of $\mathbf C_\infty$ on $\R^{2d}$ in the sense of \cite[Section~2.1]{BB}.

\paragraph{Commutation with the free dynamics.}
For $V_0=0$, the underdamped Langevin dynamics \eqref{eq:langevindy} is the Gaussian process
\begin{equation} \label{eq:freekinetic}
  \dd\begin{pmatrix}q_t\\p_t\end{pmatrix}
  =\mathbf M\begin{pmatrix}q_t\\p_t\end{pmatrix}\dd t
  +\begin{pmatrix}0\\\sqrt{2\gamma}\,\dd W_t\end{pmatrix},
  \qquad \mathbf M=\begin{pmatrix}0&\Id\\-A&-\gamma\Id\end{pmatrix}.
\end{equation}
Since $\mathbf Q_s$ commutes with $\mathbf M$, and both the Gaussian observation \eqref{eq:phasespaceobv} and the free kinetic dynamics \eqref{eq:freekinetic} preserve $\mathsf N(0,\mathbf C_\infty)$, their Gaussian transition kernels commute: in either order, the composition sends the initial state $(q,p)$ to the law
\begin{equation*}
  \mathsf N\bigl(\mathbf Q_se^{t\mathbf M}(q,p)^\top,\ \mathbf C_\infty-\mathbf Q_se^{t\mathbf M}\mathbf C_\infty e^{t\mathbf M^\top}\mathbf Q_s\bigr).
\end{equation*}
This is the reason for the common attenuation $Q_s$ in position and momentum: a block-diagonal matrix $\operatorname{diag}(Q,Q')$ commutes with $\mathbf M$ if and only if $Q'=Q$ and $Q$ commutes with $A$. Preservation of $\mathsf N(0,\mathbf C_\infty)$ then fixes the noise covariance to be the matched covariance \eqref{eq:ent-BB-block-covariance}.
Hence, for $V_0=0$ and every initial law, observing and then evolving by \eqref{eq:freekinetic} gives the same law as evolving and then observing: the bottom row of Figure~\ref{fig:kinetic-observation} is then the free kinetic dynamics, the same as the top row.

\paragraph{Observation law and conditional expectation.}
Similarly to \S~\ref{subsec:polchinski-flow}, we use the reversibility of the Gaussian observation kernel to identify the observation law and the conditional expectation that defines coarse graining. The observation \eqref{eq:phasespaceobv} takes values in the fixed coarse space $\R^{2d}$ at every scale; we call its points $(y,w)$ the \emph{coarse variables}. Since $\mathbf Q_s$ is symmetric, commutes with $\mathbf C_\infty$, and satisfies $\mathbf Q_s\mathbf C_\infty\mathbf Q_s+\mathbf C_s=\mathbf C_\infty$, the argument for \eqref{eq:gaussian-identity} gives
\begin{equation}
  \mathsf N(0,\mathbf C_\infty)(\dd q\,\dd p)\,\mathsf N\bigl(\mathbf Q_s(q,p)^\top,\mathbf C_s\bigr)(\dd y\,\dd w)
  =\mathsf N\bigl(\mathbf Q_s(y,w)^\top,\mathbf C_s\bigr)(\dd q\,\dd p)\,\mathsf N(0,\mathbf C_\infty)(\dd y\,\dd w).
  \label{eq:phase-gaussian-identity}
\end{equation}
The joint law $ \boldsymbol\pi_s$ of $(q,p)\sim\boldsymbol\mu$ and its observation $(y_s,w_s)$ is $\boldsymbol\mu(\dd q\,\dd p)\,\mathsf N\bigl(\mathbf Q_s(q,p)^\top,\mathbf C_s\bigr)(\dd y\,\dd w)$. Recalling \eqref{eq:ent-abstract-law} and multiplying \eqref{eq:phase-gaussian-identity} by $Z_0^{-1}e^{-V_0(q)}$ gives
\begin{align}
  \boldsymbol\pi_s(\dd q\,\dd p\,\dd y\,\dd w)
  &{}=Z_0^{-1}e^{-V_0(q)}\,\mathsf N(0,\mathbf C_\infty)(\dd q\,\dd p)\,\mathsf N\bigl(\mathbf Q_s(q,p)^\top,\mathbf C_s\bigr)(\dd y\,\dd w) \notag \\
  &=Z_0^{-1}e^{-V_0(q)}\,\mathsf N\bigl(\mathbf Q_s(y,w)^\top,\mathbf C_s\bigr)(\dd q\,\dd p)\,\mathsf N(0,\mathbf C_\infty)(\dd y\,\dd w), \label{eq:measurepis}
\end{align}
which extends the joint law $\pi_s$ in \eqref{eq:joint-law} to phase space. Since $V_0$ depends only on $q$ and $\mathbf C_s$ is block diagonal, integrating \eqref{eq:measurepis} over $(q,p)$ and using \eqref{eq:polchinski2} gives the observation law, which we call the \emph{renormalized phase-space measure}:
\begin{equation}
  \boldsymbol\mu_s(\dd y\,\dd w)=Z_0^{-1}e^{-V_s(Q_sy)}\,\mathsf N(0,\mathbf C_\infty)(\dd y\,\dd w)=\nu_s(\dd y)\,\mathsf N(0,\Id)(\dd w).
  \label{eq:ent-coarse-measure}
\end{equation}
In the coarse variables, the renormalized potential enters through the composition $V_s(Q_sy)$, with gradient $\nabla_y[V_s(Q_sy)]=Q_s\nabla V_s(Q_sy)$.

Bayes' formula gives the posterior
\begin{equation*}
  \boldsymbol\pi_s(\dd q\,\dd p\mid y,w)
  =e^{V_s(Q_sy)-V_0(q)}\,\mathsf N\bigl(\mathbf Q_s(y,w)^\top,\mathbf C_s\bigr)(\dd q\,\dd p).
\end{equation*}
We define the \emph{coarse-graining operator} by the conditional expectation given the observation, i.e., $\mathsf P_s\psi(y,w)=\E_{\boldsymbol\pi_s}[\psi(q,p)\mid y,w]$.
For $\psi\in L^1(\boldsymbol\mu)$ or measurable $\psi\geq0$,
\begin{equation}
  \mathsf P_s\psi(y,w)
  =e^{V_s(Q_sy)}\E_{{(\zeta_q,\zeta_p)\sim\mathsf N(0,\mathbf C_s)}}\!\left[e^{-V_0(Q_sy+\zeta_q)}\psi(Q_sy+\zeta_q,Q_sw+\zeta_p)\right],
  \label{eq:ent-coarse-operator}
\end{equation}
where $(\zeta_q,\zeta_p)$ is the Gaussian fluctuation around $\mathbf Q_s(y,w)^\top$.
On functions of position alone, $\mathsf P_s$ is the posterior average \eqref{eq:posterior-average}. The pushforward $(\mathbf Q_s)_\#\boldsymbol\mu_s$ is the renormalized measure of Bauerschmidt and Bodineau, and $\psi\mapsto(\mathsf P_s\psi)\circ\mathbf Q_s^{-1}$ is their Polchinski semigroup; see Appendix~\ref{app:BB-identification} for the details.

Let $f$ be a probability density relative to $\boldsymbol\mu$ and write $f_s=\mathsf P_sf$. If $(q,p)$ has law $f\boldsymbol\mu$, its observation $(y_s,w_s)$ has law $f_s\boldsymbol\mu_s$, that is, the coarse-graining operator maps the density of a law relative to $\boldsymbol\mu$ to the density of its observation law relative to $\boldsymbol\mu_s$. Indeed, the joint law of $(q,p)$ and $(y_s,w_s)$ is $f(q,p)\,\boldsymbol\pi_s(\dd q\,\dd p\,\dd y\,\dd w)$. Its $(q,p)$-marginal is $f\boldsymbol\mu$, and since $\mathsf P_s$ is the conditional expectation given the observation, its $(y,w)$-marginal is $f_s\boldsymbol\mu_s$. As a conditional expectation, $\mathsf P_s$ fixes constants, and $f_s\boldsymbol\mu_s$ is a probability law:
\begin{equation}
  \mathsf P_s1=1,\qquad \int f_s\dd\boldsymbol\mu_s=\int f\dd\boldsymbol\mu=1.
  \label{eq:ent-coarse-mass}
\end{equation}
At $s=0$, $\mathbf Q_0=\Id$ and $\mathbf C_0=0$, so $\boldsymbol\mu_0=\boldsymbol\mu$ and $f_0=f$. We call the family $s\mapsto(\boldsymbol\mu_s,f_s)$ the \emph{phase-space Polchinski flow}.
For the underdamped Langevin dynamics, we apply this construction at each kinetic time. The observation of $(q_t,p_t)\sim\boldsymbol\varrho_t$ is $(y_{t,s},w_{t,s})^\top=\mathbf Q_s(q_t,p_t)^\top+\boldsymbol\xi_s$, with $\boldsymbol\xi_s\sim\mathsf N(0,\mathbf C_s)$ independent of the dynamics. Then, the density of $(y_{t,s},w_{t,s})$ relative to $\boldsymbol\mu_s$ and its law are
\begin{equation}
  g_{t,s}=\mathsf P_sg_t,\qquad
  \boldsymbol\varrho_{t,s}(\dd y\,\dd w)=g_{t,s}(y,w)\,\boldsymbol\mu_s(\dd y\,\dd w).
  \label{eq:ent-observed-law}
\end{equation}

\subsection{Scale evolution and entropy decomposition}
\label{subsec:scale-generator}

As the scale $s \to \infty$, the noisy observation forgets the state, and the relative entropy $\Ent_{\boldsymbol\mu_s}(f_s)$ between the observation laws of $f\boldsymbol\mu$ and $\boldsymbol\mu$ decreases from $\Ent_{\boldsymbol\mu}(f)$ to zero. We will compute its rate of decrease in Lemma~\ref{lem:ent-scale-generator}, by the scale evolution of $f_s$ and $\boldsymbol\mu_s$, which gives a Fisher information in position and momentum. Integrating this rate over all scales then gives the entropy decomposition (see Lemma~\ref{lem:ent-decomposition}). In the proofs below, we mainly consider \emph{regular densities} relative to $\boldsymbol\mu$, satisfying
\begin{equation}
  f\in C_b^\infty(\R^{2d}),\qquad \E_{\boldsymbol\mu} f=1,\qquad \inf f>0.
  \label{eq:ent-regular-class}
\end{equation}
The extension to general densities is discussed in Appendix~\ref{app:regularity}.

\paragraph{Scores and Fisher information.}
Let $f$ be a regular density. The \emph{position and momentum scores} of $f_s=\mathsf P_sf$ are its logarithmic gradients $\nabla_y\log f_s$ and $\nabla_w\log f_s$. The corresponding \emph{position}, \emph{momentum}, and \emph{$A$-weighted momentum} Fisher informations of $f_s$ relative to $\boldsymbol\mu_s$ are defined by
\begin{equation}
  \begin{gathered}
  \mathcal I^y_s=\int f_s\abs{\nabla_y\log f_s}^2\dd\boldsymbol\mu_s, \qquad \mathcal I^w_s=\int f_s\abs{\nabla_w\log f_s}^2\dd\boldsymbol\mu_s,\\
  \mathcal I^{w,A}_s=\int f_s(\nabla_w\log f_s)^\top A\,\nabla_w\log f_s\dd\boldsymbol\mu_s,
  \end{gathered}
  \label{eq:ent-scale-scores}
\end{equation}
respectively.
We write $\mathcal I^y_s(f_s)$, $\mathcal I^w_s(f_s)$, and $\mathcal I^{w,A}_s(f_s)$ when the density needs to be specified. Since $f_0=f$ and $\boldsymbol\mu_0=\boldsymbol\mu$, we have $\mathcal I^w_0=\mathcal I_p(f)$ (see \eqref{eq:fishermoment}). The assumption $A\succeq\Id$ gives $\mathcal I^{w,A}_s\geq\mathcal I^w_s$.

\begin{lemma}[Polchinski generator and renormalized phase-space measure]
\label{lem:ent-scale-generator}
Define
\begin{equation}
  \mathsf L_s =\frac12\Delta_y +\frac12A:\nabla_w^2 -\left(\frac12Ay+Q_s\nabla V_s(Q_sy)\right)\cdot\nabla_y -\frac12Aw\cdot\nabla_w.
  \label{eq:ent-scale-generator}
\end{equation}
For $\psi\in C_b^\infty(\R^{2d})$, $\partial_s(\mathsf P_s\psi)=\mathsf L_s(\mathsf P_s\psi)$. In particular, for a regular density $f$ and $f_s = \mathsf P_s f$,
\begin{equation}
  \partial_sf_s=\mathsf L_sf_s.
  \label{eq:ent-scale-evolution}
\end{equation}
The renormalized phase-space measure satisfies, in the sense of distributions,
\begin{equation}\label{eq:der_mus}
  \frac\dd{\dd s}\boldsymbol\mu_s
  =-\partial_s[V_s(Q_sy)]\,\boldsymbol\mu_s
  =-\frac12\diver_y\bigl(Q_s\nabla V_s(Q_sy)\,\boldsymbol\mu_s\bigr).
\end{equation}
If $\psi\in C^{1,2}([0,\infty)\times\R^{2d})$ and $\partial_s\psi_s$, $\nabla_{(y,w)}^k\psi_s$ $(0\leq k\leq2)$ have polynomial growth locally uniformly in $s$, then
\begin{equation}
  \frac\dd{\dd s}\int \psi_s\dd\boldsymbol\mu_s =\int(\partial_s\psi_s-\mathsf L_s\psi_s)\dd\boldsymbol\mu_s.
  \label{eq:ent-dual-scale}
\end{equation}
\end{lemma}

Lemma~\ref{lem:ent-scale-generator} is \cite[Proposition~2.1]{BB} (see also \cite[Proposition~3.5]{BBD2024}) applied to the covariance decomposition \eqref{eq:ent-BB-block-covariance} of $\mathbf C_\infty$ with interaction $(q,p)\mapsto V_0(q)$, and written in the coarse variables. Appendix~\ref{app:BB-identification} carries out the change of variables through $\mathbf Q_s$ and identifies the measure, semigroup, generator, and entropy decomposition with those of \cite{BB}.

\begin{proof}
The coarse-graining operator \eqref{eq:ent-coarse-operator} and the
density $Z_0^{-1}e^{-V_s(Q_sy)}$ of $\boldsymbol\mu_s$ relative to
$\mathsf N(0,\mathbf C_\infty)$ are built from Gaussian averages against
the kernels $\mathsf N(\mathbf Q_s(y,w)^\top,\mathbf C_s)$, $s\geq0$.
Since $\mathbf Q_{s+u}=\mathbf Q_s\mathbf Q_u$ and
$\mathbf C_{s+u}=\mathbf C_s+\mathbf Q_s\mathbf C_u\mathbf Q_s$, these
kernels form a Markov semigroup. Its generator has drift
$\frac\dd{\dd s}\mathbf Q_s\big|_{s=0}(y,w)^\top
=-\frac12\operatorname{diag}(A,A)(y,w)^\top$, since $Q_s=e^{-sA/2}$, and
diffusion matrix $\frac\dd{\dd s}\mathbf C_s\big|_{s=0}=\operatorname{diag}(\Id,A)$
by \eqref{eq:ent-BB-block-covariance}, so it is the Ornstein--Uhlenbeck
operator
\[
  \mathcal G=\frac12\Delta_y+\frac12A:\nabla_w^2
       -\frac12Ay\cdot\nabla_y-\frac12Aw\cdot\nabla_w.
\]
The Gaussian measure $\mathsf N(0,\mathbf C_\infty)$ is invariant for
$\mathcal G$.
We will compute the scale derivatives in $s$ of the coarse-graining operator and of the density of $\boldsymbol\mu_s$ through $\mathcal G$.

We first relate $\mathsf L_s$ to $\mathcal G$. Since $e^{-V_s(Q_sy)}$ depends only on $y$ and
$\nabla_y e^{-V_s(Q_sy)}=-e^{-V_s(Q_sy)}Q_s\nabla V_s(Q_sy)$, the product
rule for $\mathcal G$ gives, for $\psi\in C^2(\R^{2d})$, the
identity
\begin{equation}
  \mathsf L_s\psi
  = e^{V_s(Q_sy)} \mathcal G\bigl(e^{-V_s(Q_sy)}\psi\bigr)
       - e^{V_s(Q_sy)} \psi\,\mathcal G\bigl(e^{-V_s(Q_sy)}\bigr),
  \label{eq:ent-normalized-OU}
\end{equation}
that is, $\mathsf L_s$ is the Doob transform of $\mathcal G$ by
$e^{-V_s(Q_sy)}$.

The equation for $\mathsf P_s\psi$ follows by writing it as a ratio.
Let $\psi\in C_b^\infty(\R^{2d})$. By \eqref{eq:ent-coarse-operator}
and \eqref{eq:polchinski2} evaluated at $Q_sy$,
$\mathsf P_s\psi=n_s/e^{-V_s(Q_sy)}$, where
\[
  n_s(y,w)=\E\bigl[e^{-V_0(Q_sy+\zeta_q)}
                    \psi(Q_sy+\zeta_q,Q_sw+\zeta_p)\bigr],
  \qquad
  e^{-V_s(Q_sy)}=\E\bigl[e^{-V_0(Q_sy+\zeta_q)}\bigr],
\]
with $(\zeta_q,\zeta_p)\sim\mathsf N(0,\mathbf C_s)$. As Gaussian
averages of the functions $e^{-V_0(q)}\psi(q,p)$ and $e^{-V_0(q)}$, which
belong to $C_b^\infty(\R^{2d})$, the numerator and the denominator both
solve the Ornstein--Uhlenbeck equation
\begin{equation}
  \partial_sn_s=\mathcal G n_s,\qquad
  \partial_s e^{-V_s(Q_sy)}=\mathcal G\bigl(e^{-V_s(Q_sy)}\bigr).
  \label{eq:ent-OU-equation}
\end{equation}
Since $n_s=e^{-V_s(Q_sy)}\mathsf P_s\psi$, the quotient rule and
\eqref{eq:ent-normalized-OU} give
\begin{align*}
  \partial_s(\mathsf P_s\psi)
  &=e^{V_s(Q_sy)}\bigl(\partial_sn_s
     -\mathsf P_s\psi\,\partial_se^{-V_s(Q_sy)}\bigr)\\
  &=e^{V_s(Q_sy)}\Bigl(\mathcal G\bigl(e^{-V_s(Q_sy)}\mathsf P_s\psi\bigr)
     -\mathsf P_s\psi\,\mathcal G\bigl(e^{-V_s(Q_sy)}\bigr)\Bigr)
  =\mathsf L_s(\mathsf P_s\psi).
\end{align*}
This proves \eqref{eq:ent-scale-evolution} with $\psi=f$.

We now derive the measure equation \eqref{eq:der_mus}, by testing against $\psi\in C_c^\infty(\R^{2d})$. The first equality
in \eqref{eq:der_mus} is the derivative of the density
$Z_0^{-1}e^{-V_s(Q_sy)}$ of $\boldsymbol\mu_s$ relative to
$\mathsf N(0,\mathbf C_\infty)$ (see \eqref{eq:ent-coarse-measure}). For the second, we integrate
$\mathcal G$ by parts. Since $\mathsf N(0,\mathbf C_\infty)$
has density proportional to $e^{-\frac12\ip y{Ay}-\frac12|w|^2}$ and $\psi$ has compact support,
Gaussian integration by parts gives, for every $\chi\in C^2(\R^{2d})$,
\begin{equation}
  \int\psi\,\mathcal G\chi\dd\mathsf N(0,\mathbf C_\infty)
  =-\frac12\int\bigl(\nabla_y\psi\cdot\nabla_y\chi
     +(\nabla_w\psi)^\top A\nabla_w\chi\bigr)\dd\mathsf N(0,\mathbf C_\infty).
  \label{eq:ent-OU-ibp}
\end{equation}
For $\chi=e^{-V_s(Q_sy)}$, only the position gradient is involved, and hence \eqref{eq:ent-OU-equation} gives
\begin{align*}
  \frac\dd{\dd s}\int\psi\dd\boldsymbol\mu_s
  &=Z_0^{-1}\int\psi\,\mathcal G\bigl(e^{-V_s(Q_sy)}\bigr)
     \dd\mathsf N(0,\mathbf C_\infty)
  =-\frac1{2Z_0}\int\nabla_y\psi\cdot\nabla_ye^{-V_s(Q_sy)}
     \dd\mathsf N(0,\mathbf C_\infty)\\
  &=\frac12\int Q_s\nabla V_s(Q_sy)\cdot\nabla_y\psi
                         \dd\boldsymbol\mu_s,
\end{align*}
which is the second equality tested against $\psi$.

Finally, we integrate \eqref{eq:ent-normalized-OU} against
$\mathsf N(0,\mathbf C_\infty)$. By the invariance of
$\mathsf N(0,\mathbf C_\infty)$ for $\mathcal G$, the first term on the right-hand side integrates to zero, and \eqref{eq:ent-OU-equation} gives
\[
  \int\mathsf L_s\psi_s\dd\boldsymbol\mu_s
  =-Z_0^{-1}\int\psi_s\,\partial_s e^{-V_s(Q_sy)}
                        \dd\mathsf N(0,\mathbf C_\infty).
\]
Differentiating under the integral and inserting this identity,
\[
  \frac\dd{\dd s}\int\psi_s\dd\boldsymbol\mu_s
  =Z_0^{-1}\int\bigl(e^{-V_s(Q_sy)}\partial_s\psi_s
     +\psi_s\,\partial_se^{-V_s(Q_sy)}\bigr)\dd\mathsf N(0,\mathbf C_\infty)
  =\int(\partial_s\psi_s-\mathsf L_s\psi_s)\dd\boldsymbol\mu_s,
\]
which is \eqref{eq:ent-dual-scale}.
\end{proof}

\paragraph{Carr\'e du champ of the Polchinski generator.}
The \emph{carr\'e du champ} of $\mathsf L_s$, with the convention of \cite{BGL}, is the bilinear form
\begin{equation*}
  \Gamma_s(\psi,\chi)=\frac12\bigl[\mathsf L_s(\psi\chi)-\psi\,\mathsf L_s\chi-\chi\,\mathsf L_s\psi\bigr],\qquad \Gamma_s(\psi)=\Gamma_s(\psi,\psi).
\end{equation*}
Note from \eqref{eq:ent-scale-generator} that only the second-order part $\frac12\Delta_y+\frac12A:\nabla_w^2$ of $\mathsf L_s$ contributes to $\Gamma_s$. We have
\begin{equation}
  \Gamma_s(\psi,\chi)=\frac12\nabla_y\psi\cdot\nabla_y\chi+\frac12(\nabla_w\psi)^\top A\nabla_w\chi,
  \label{eq:ent-carre-du-champ}
\end{equation}
and the chain rule
\begin{equation} \label{eq:chainrule}
  \mathsf L_s[a(\psi)]=a'(\psi)\,\mathsf L_s\psi+a''(\psi)\,\Gamma_s(\psi)
\end{equation}
holds for every smooth function $a\colon\R\to\R$. Since $\Gamma_s(f_s)/f_s=f_s\Gamma_s(\log f_s)$, the formula \eqref{eq:ent-carre-du-champ} and the definitions \eqref{eq:ent-scale-scores} give
\begin{equation} \label{auxeq:ident}
  \int\frac{\Gamma_s(f_s)}{f_s}\dd\boldsymbol\mu_s=\int f_s\Gamma_s(\log f_s)\dd\boldsymbol\mu_s=\frac12\bigl(\mathcal I^y_s+\mathcal I^{w,A}_s\bigr).
\end{equation}
This is the Fisher information of $f_s$ relative to $\boldsymbol\mu_s$ associated with $\Gamma_s$.

\begin{lemma}[Entropy decomposition]
\label{lem:ent-decomposition}
Let $f$ be a regular density. Then
\begin{equation}
  \frac\dd{\dd s}\Ent_{\boldsymbol\mu_s}(f_s) =-\int\frac{\Gamma_s(f_s)}{f_s}\dd\boldsymbol\mu_s =-\frac12\bigl(\mathcal I^y_s+\mathcal I^{w,A}_s\bigr),
  \label{eq:ent-decomposition-rate}
\end{equation}
and hence
\begin{equation}
  \Ent_{\boldsymbol\mu}(f)=\frac12\int_0^\infty\bigl(\mathcal I^y_s+\mathcal I^{w,A}_s\bigr)\dd s.
  \label{eq:ent-decomposition}
\end{equation}
\end{lemma}

\begin{proof}
By Lemma~\ref{lem:ent-fixed-cutoff-regularity}, $\psi_s=f_s\log f_s$ satisfies the growth condition of Lemma~\ref{lem:ent-scale-generator}. By \eqref{eq:ent-dual-scale}, \eqref{eq:ent-scale-evolution}, and the chain rule \eqref{eq:chainrule} with $a(u)=u\log u$,
\begin{align*}
  \frac\dd{\dd s}\Ent_{\boldsymbol\mu_s}(f_s) &=\int(\partial_s\psi_s-\mathsf L_s\psi_s)\dd\boldsymbol\mu_s\\
  &=\int\bigl[(1+\log f_s)\mathsf L_sf_s-\mathsf L_s(f_s\log f_s)\bigr]\dd\boldsymbol\mu_s =-\int\frac{\Gamma_s(f_s)}{f_s}\dd\boldsymbol\mu_s.
\end{align*}
The identity \eqref{auxeq:ident} gives the second equality in \eqref{eq:ent-decomposition-rate}. Integrating \eqref{eq:ent-decomposition-rate} over $[0,\infty)$, with $\Ent_{\boldsymbol\mu_0}(f_0)=\Ent_{\boldsymbol\mu}(f)$ and $\Ent_{\boldsymbol\mu_s}(f_s)\to0$ as $s\to\infty$ (Lemma~\ref{lem:ent-scale-endpoints}), proves \eqref{eq:ent-decomposition}.
\end{proof}

\subsection{The force defect and the coarse kinetic equation}
\label{subsec:matched-kinetic-equation}

In this subsection, we identify the evolution of the observation law in kinetic time (the lower row of Figure~\ref{fig:kinetic-observation}). For the free dynamics, the matched observation of \S~\ref{subsec:phase-construction} makes the diagram commute. With an interaction, the observation law evolves by the density equation of the underdamped Langevin dynamics \eqref{eq:langevindy} in the coarse variables, with $V_0$ replaced by $V_s(Q_sy)$, plus one new term: the momentum divergence of the force defect, which measures how far the interaction force averaged over the unresolved fluctuations is from its equilibrium value; see Lemma~\ref{lem:ent-intertwining}.

\paragraph{The force defect.}
Let $(q,p)$ have law $f\boldsymbol\mu$ with a regular density $f$. By \S~\ref{subsec:phase-construction}, the joint law of $(q,p)$ and its observation is $f(q,p)\,\boldsymbol\pi_s$, with $\boldsymbol\pi_s$ as in \eqref{eq:measurepis} and $(y,w)$-marginal $f_s\boldsymbol\mu_s$. Hence the posterior of $(q,p)$ given $(y_s,w_s)=(y,w)$ is the equilibrium posterior weighted by $f/f_s$,
\begin{equation}
  \frac{f(q,p)}{f_s(y,w)}\,\boldsymbol\pi_s(\dd q\,\dd p\mid y,w),
  \label{eq:ent-weighted-posterior}
\end{equation}
and the conditional interaction force under $f\boldsymbol\mu$ is
\begin{equation*}
  \E\bigl[\nabla V_0(q)\mid(y_s,w_s)=(y,w)\bigr]
  =\int\nabla V_0(q)\,\frac{f(q,p)}{f_s(y,w)}\,\boldsymbol\pi_s(\dd q\,\dd p\mid y,w)
  =\frac{\mathsf P_s(f\nabla V_0)(y,w)}{f_s(y,w)},
\end{equation*}
where $(q,p)\sim f\boldsymbol\mu$ and the last equality holds because $\mathsf P_s$ is the conditional expectation under $\boldsymbol\pi_s$. At equilibrium, where $f=1$, this force is $\mathsf P_s(\nabla V_0)$, and differentiating \eqref{eq:polchinski2} in $q$ and comparing with \eqref{eq:posterior-average} gives $\mathsf P_s(\nabla V_0)=\nabla V_s(Q_sy)$. We call $\nabla V_s(Q_sy)$, the gradient of the renormalized potential at the coarse field, the \emph{renormalized force}. The \emph{force defect} weights the difference of the two conditional forces by $f_s$ and transfers it to the coarse variable $y$ by $Q_s$, similarly to $\nabla_y[V_s(Q_sy)]=Q_s\nabla V_s(Q_sy)$:
\begin{equation}
  R_s=Q_s\bigl[\mathsf P_s(f\nabla V_0)-f_s\,\nabla V_s(Q_sy)\bigr]
  =Q_s\bigl[\mathsf P_s(f\nabla V_0)-\mathsf P_sf\,\mathsf P_s(\nabla V_0)\bigr].
  \label{eq:ent-force-defect}
\end{equation}
By the second form of \eqref{eq:ent-force-defect}, $R_s$ is $Q_s$ applied to the covariance of $f$ and $\nabla V_0$ under the posterior $\boldsymbol\pi_s(\dd q\,\dd p\mid y,w)$; in particular, $R_s=0$ if $f=1$ or $V_0=0$.

\paragraph{Coarse graining and derivatives.}
Coarse graining commutes with every derivative in the kinetic generator except the position gradient, and the defect of the position gradient is the force defect. Let $\nabla_q^*=-\diver_q+\nabla U\cdot$ and $\nabla_p^*=-\diver_p+p\cdot$ be the adjoints of $\nabla_q$ and $\nabla_p$ in $L^2(\boldsymbol\mu)$, and
\begin{equation*}
  \nabla_y^*=-\diver_y+\bigl(Ay+Q_s\nabla V_s(Q_sy)\bigr)\cdot,\qquad
  \nabla_w^*=-\diver_w+w\cdot
\end{equation*}
be the adjoints of $\nabla_y$ and $\nabla_w$ in $L^2(\boldsymbol\mu_s)$, recalling \eqref{eq:ent-abstract-law} and \eqref{eq:ent-coarse-measure}. The adjoint of $\mathsf P_s\colon L^2(\boldsymbol\mu)\to L^2(\boldsymbol\mu_s)$ is the conditional expectation given the state, $\mathsf P_s^*\psi(q,p)=\E_{\boldsymbol\xi_s}\psi\bigl(\mathbf Q_s(q,p)^\top+\boldsymbol\xi_s\bigr)$, so that $\nabla_q\mathsf P_s^*\psi=Q_s\mathsf P_s^*(\nabla_y\psi)$ and $\nabla_p\mathsf P_s^*\psi=Q_s\mathsf P_s^*(\nabla_w\psi)$. By duality,
\begin{equation}
  \mathsf P_s\nabla_q^*=\nabla_y^*Q_s\mathsf P_s,\qquad
  \mathsf P_s\nabla_p^*=\nabla_w^*Q_s\mathsf P_s.
  \label{eq:ent-adjoint-commutation}
\end{equation}
The corresponding relations for the gradients follow by differentiating \eqref{eq:ent-coarse-operator} with \eqref{eq:ent-force-defect}:
\begin{equation}
  \nabla_w \mathsf P_s f =Q_s\mathsf P_s(\nabla_pf),\qquad
  \nabla_y \mathsf P_s f = Q_s\mathsf P_s(\nabla_qf)-R_s.
  \label{eq:ent-gradient-defect}
\end{equation}
We hence see that the interaction force enters the kinetic generator through $\nabla_q^*$, which passes through coarse graining exactly because $\mathsf P_s^*$ is an unweighted Gaussian average, while the gradient $\nabla_q$ meets the weight $e^{-V_0}$ of the posterior average \eqref{eq:ent-coarse-operator}, and the resulting mismatch between $\nabla V_0$ and $\nabla V_s(Q_sy)$ is the force defect.

\paragraph{The coarse kinetic equation.}
Note from \eqref{eq:ent-coarse-measure} that $\boldsymbol\mu_s$ is the Gibbs measure of the potential $\frac12\ip y{Ay}+V_s(Q_sy)+\frac12|w|^2$.
We introduce the corresponding Hamiltonian transport operator and the momentum Ornstein--Uhlenbeck operator on the coarse space by, respectively,
\begin{equation}
  \mathcal A_s =w\cdot\nabla_y-\bigl(Ay+Q_s\nabla V_s(Q_sy)\bigr)\cdot\nabla_w, \qquad \mathcal S_w=\Delta_w-w\cdot\nabla_w,
  \label{eq:ent-coarse-kinetic-generator}
\end{equation}
as in \eqref{eq:ent-abstract-generator}. In divergence form, we can write
\begin{equation} \label{eq:divergenceform}
  \mathcal A=\nabla_p^*\nabla_q-\nabla_q^*\nabla_p,\qquad
  \mathcal S=-\nabla_p^*\nabla_p,\qquad
  \mathcal A_s=\nabla_w^*\nabla_y-\nabla_y^*\nabla_w,\qquad
  \mathcal S_w=-\nabla_w^*\nabla_w.
\end{equation}
In particular, $\mathcal A_s^*=-\mathcal A_s$ and $\mathcal S_w^*=\mathcal S_w$ in $L^2(\boldsymbol\mu_s)$.

For $g_t=\mathcal P_t^*g_0$ and $g_{t,s}=\mathsf P_sg_t$ as in \eqref{eq:ent-observed-law}, the force defect \eqref{eq:ent-force-defect} of $f=g_t$ is
\begin{equation*}
  R_{t,s}=Q_s\bigl[\mathsf P_s(g_t\nabla V_0)-g_{t,s}\,\nabla V_s(Q_sy)\bigr].
\end{equation*}
We also call $r_{t,s}=R_{t,s}/g_{t,s}$ the \emph{normalized force defect}. Since $\mathcal S_w=-\nabla_w^*\nabla_w$, the momentum divergence $-\nabla_w^*R_{t,s}=-\nabla_w^*(g_{t,s}r_{t,s})$ then has the same form as the friction term $\gamma\mathcal S_wg_{t,s}=-\nabla_w^*(g_{t,s}\,\gamma\nabla_w\log g_{t,s})$.

\begin{lemma}[Matched kinetic intertwining]
\label{lem:ent-intertwining}
Let $g_0$ be a regular density and $g_t=\mathcal P_t^*g_0$. Then, for every $s\geq0$,
\begin{equation}
  \partial_tg_{t,s} =-\mathcal A_sg_{t,s} +\gamma\mathcal S_wg_{t,s} -\nabla_w^*R_{t,s}.
  \label{eq:ent-coarse-density-equation}
\end{equation}
\end{lemma}

\begin{proof}
By Lemma~\ref{lem:ent-kinetic-regularity}, $g_t$ is a regular density for every $t\geq0$. For a regular density $f$, by the divergence form \eqref{eq:divergenceform} and the identities \eqref{eq:ent-adjoint-commutation} and \eqref{eq:ent-gradient-defect}, we have
\begin{align*}
  \mathsf P_s(\mathcal Af)
  &=\nabla_w^*Q_s\mathsf P_s(\nabla_qf)-\nabla_y^*Q_s\mathsf P_s(\nabla_pf)
   =\nabla_w^*(\nabla_yf_s+R_s)-\nabla_y^*\nabla_wf_s
   =\mathcal A_sf_s+\nabla_w^*R_s,\\
  \mathsf P_s(\mathcal Sf)
  &=-\nabla_w^*Q_s\mathsf P_s(\nabla_pf)
   =-\nabla_w^*\nabla_wf_s=\mathcal S_wf_s.
\end{align*}
Since $\mathsf P_s$ does not depend on $t$, we have $\partial_tg_{t,s}=\mathsf P_s(\partial_tg_t)$. Applying the two identities above with $f=g_t$ to the density equation $\partial_tg_t=-\mathcal Ag_t+\gamma\mathcal Sg_t$ following \eqref{eq:ent-evolving-law}, we obtain
\begin{equation*}
  \partial_tg_{t,s}
  =-\mathsf P_s(\mathcal Ag_t)+\gamma\mathsf P_s(\mathcal Sg_t)
  =-\mathcal A_sg_{t,s}+\gamma\mathcal S_wg_{t,s}-\nabla_w^*R_{t,s}.\qedhere
\end{equation*}
\end{proof}

\begin{remark}[Probabilistic derivation]
\label{rem:ent-intertwining-probabilistic}
Equation \eqref{eq:ent-coarse-density-equation} can also be derived by conditioning on the observation. Fix $s\geq0$ and $\psi\in C_c^\infty(\R^{2d})$, and write
\begin{equation*}
  \mathcal L_{\mathrm{free}}=w\cdot\nabla_y-Ay\cdot\nabla_w+\gamma\mathcal S_w
\end{equation*}
for the free kinetic generator in the coarse variables. Differentiating at $t=0$ the commutation of the Gaussian observation with the free kinetic dynamics \eqref{eq:freekinetic} shows that the generator \eqref{eq:ent-abstract-generator} with $V_0=0$ maps $\mathsf P_s^*\psi$ to $\mathsf P_s^*(\mathcal L_{\mathrm{free}}\psi)$, and $\nabla_p\mathsf P_s^*\psi=Q_s\mathsf P_s^*(\nabla_w\psi)$. Since the interaction part of \eqref{eq:ent-abstract-generator} is $-\nabla V_0(q)\cdot\nabla_p$, the observation $(y_{t,s},w_{t,s})^\top=\mathbf Q_s(q_t,p_t)^\top+\boldsymbol\xi_s$ satisfies
\begin{equation*}
  \frac\dd{\dd t}\int\psi\dd\boldsymbol\varrho_{t,s}
  =\int\mathcal L_{\mathrm{free}}\psi\dd\boldsymbol\varrho_{t,s}
   -\E\bigl[Q_s\nabla V_0(q_t)\cdot\nabla_w\psi(y_{t,s},w_{t,s})\bigr],
\end{equation*}
where the expectation is over the state and its independent observation noise. Conditioning on the observation and using \eqref{eq:ent-weighted-posterior} with $f=g_t$,
\begin{equation*}
  \E\bigl[Q_s\nabla V_0(q_t)\mid(y_{t,s},w_{t,s})=(y,w)\bigr]
  =\frac{Q_s\mathsf P_s(g_t\nabla V_0)(y,w)}{g_{t,s}(y,w)}
  =Q_s\nabla V_s(Q_sy)+r_{t,s}(y,w).
\end{equation*}
By \eqref{eq:ent-coarse-kinetic-generator}, the renormalized force combines with $\mathcal L_{\mathrm{free}}$ into $\mathcal A_s+\gamma\mathcal S_w$, and the remainder gives the force defect $R_{t,s}=g_{t,s}r_{t,s}$:
\begin{equation*}
  \frac\dd{\dd t}\int\psi\,g_{t,s}\dd\boldsymbol\mu_s
  =\int(\mathcal A_s+\gamma\mathcal S_w)\psi\,g_{t,s}\dd\boldsymbol\mu_s
   -\int\nabla_w\psi\cdot R_{t,s}\dd\boldsymbol\mu_s.
\end{equation*}
Since $\boldsymbol\mu_s$ does not depend on $t$, the adjoint relations of $\mathcal A_s$, $\mathcal S_w$, and $\nabla_w^*$ turn this weak form into \eqref{eq:ent-coarse-density-equation}. In this derivation the force defect appears as the difference between the conditional interaction forces; in the proof of Lemma~\ref{lem:ent-intertwining}, it appears as the defect of the position gradient under coarse graining \eqref{eq:ent-gradient-defect}.
\end{remark}

\section{The corrector and entropy decay}
\label{sec:modified-entropy}

This section is devoted to the proofs of Theorem~\ref{thm:phase-space-criterion} and Corollary~\ref{cor:abstract-L2}. We add to the entropy a corrector, built from the entropy decomposition of Lemma~\ref{lem:ent-decomposition}, and reduce both results to a bound on the kinetic derivative of the corrector (Proposition~\ref{prop:ent-corrector-differential}). This bound rests on the exact identity for the kinetic derivative of the cross term at each scale (Lemma~\ref{lem:ent-exact-cross-evolution}), and two energy estimates along the Polchinski flow, one carrying the control of the momentum Fisher information by friction from scale zero to all scales (Lemma~\ref{lem:ent-velocity-fisher}), the other bounding the force defect by the entropy (Proposition~\ref{prop:ent-force-energy}).

Throughout, $g_0$ is a regular density, that is, \eqref{eq:ent-regular-class} holds with $f=g_0$; we write $g_t=\mathcal P_t^*g_0$ and $g_{t,s}=\mathsf P_sg_t$. By Lemma~\ref{lem:ent-kinetic-regularity}, $g_t$ is a regular density for every $t\geq0$, so at each fixed kinetic time the estimates for a fixed density apply with $f=g_t$ and $f_s=g_{t,s}$. At each fixed scale $s$, the coarse-grained density $g_{t,s}$ evolves in kinetic time by the coarse kinetic equation \eqref{eq:ent-coarse-density-equation}, relative to the reference measure $\boldsymbol\mu_s$, which does not depend on $t$.

\subsection{The corrector and the proof of \texorpdfstring{Theorem~\ref{thm:phase-space-criterion}}{the phase-space criterion}}
\label{subsec:corrector-size}
Following the modified entropy method of hypocoercivity~\cite{Villani}, we add to the entropy a term coupling the position and momentum scores, which captures the transfer of dissipation from momentum to position by Hamiltonian transport.

Specifically, let $f$ be a regular density and $f_s=\mathsf P_sf$. At each scale $s$, we define the \emph{cross term} by
\begin{equation*}
  \mathcal J_s(f_s)=\int f_s\,\nabla_y\log f_s\cdot\nabla_w\log f_s\dd\boldsymbol\mu_s,
\end{equation*}
and the \emph{corrector} by integrating the cross term over all scales:
\begin{equation}
  \mathcal J(f)=\int_0^\infty\mathcal J_s(f_s)\dd s,
  \label{eq:ent-corrector}
\end{equation}
as the entropy decomposition of Lemma~\ref{lem:ent-decomposition} integrates Fisher informations over scales. Then, we introduce the \emph{modified entropy}, for $0<\theta<1$,
\begin{equation}
  \mathcal E_\theta(f)=\Ent_{\boldsymbol\mu}(f)+\theta\,\mathcal J(f).
  \label{eq:ent-Lyapunov}
\end{equation}
We first bound the size of the corrector in Lemma~\ref{lem:ent-corrector-size}, which gives the equivalence between the modified entropy $\mathcal E_\theta(f)$ and the entropy $\Ent_{\boldsymbol\mu}(f)$:
\begin{equation}
  (1-\theta)\Ent_{\boldsymbol\mu}(f)\leq\mathcal E_\theta(f)\leq(1+\theta)\Ent_{\boldsymbol\mu}(f).
  \label{eq:ent-Lyapunov-equivalence}
\end{equation}

\begin{lemma}[Size of the corrector]
\label{lem:ent-corrector-size}
For a regular density $f$, the integral \eqref{eq:ent-corrector} is well-defined, and satisfies
\begin{equation}
  \abs{\mathcal J(f)}\leq \Ent_{\boldsymbol\mu}(f).
  \label{eq:ent-corrector-size}
\end{equation}
\end{lemma}

\begin{proof}
By Young's inequality and $A\succeq\Id$, the integrand of $\mathcal J_s(f_s)$ obeys the pointwise bound
\begin{equation*}
  2\bigl|\nabla_y\log f_s\cdot\nabla_w\log f_s\bigr|\leq|\nabla_y\log f_s|^2+|\nabla_w\log f_s|^2\leq|\nabla_y\log f_s|^2+(\nabla_w\log f_s)^\top A\,\nabla_w\log f_s.
\end{equation*}
Multiplying by $f_s$, integrating against $\boldsymbol\mu_s$ and over $s$, and using \eqref{eq:ent-decomposition} gives
\begin{equation*}
  \int_0^\infty\!\int f_s\bigl|\nabla_y\log f_s\cdot\nabla_w\log f_s\bigr|\dd\boldsymbol\mu_s\dd s\leq\frac12\int_0^\infty\bigl(\mathcal I^y_s+\mathcal I^{w,A}_s\bigr)\dd s=\Ent_{\boldsymbol\mu}(f).
\end{equation*}
This proves \eqref{eq:ent-corrector-size}, completing the proof.
\end{proof}

We now prove Theorem~\ref{thm:phase-space-criterion} by showing that the modified entropy $\mathcal E_\theta(g_t)$ decays exponentially for a suitable weight $\theta$, and deduce Corollary~\ref{cor:abstract-L2} by linearization.  By \eqref{eq:ent-physical-dissipation}, the friction part dissipates entropy only through the momentum Fisher information. The time derivative of $\mathcal J(g_t)$ provides the missing dissipation in the position variable.

\begin{proposition}[Bound on the kinetic derivative of the corrector]
\label{prop:ent-corrector-differential}
Let $g_t=\mathcal P_t^*g_0$ with $g_0$ a regular density. Then
\begin{equation}
  \begin{aligned}
  \frac\dd{\dd t}\mathcal J(g_t) \leq{}& -2(1-\eta)\Ent_{\boldsymbol\mu}(g_t) +\left(\max\{2,\gamma\}+\frac{\eta}{\sqrt2}\right)\mathcal I_p(g_t)\\
  &+\sqrt2(\gamma+\eta)\sqrt{\mathcal I_p(g_t)\,\Ent_{\boldsymbol\mu}(g_t)}.
  \end{aligned}
  \label{eq:ent-corrector-differential}
\end{equation}
\end{proposition}
The first term on the right-hand side of \eqref{eq:ent-corrector-differential} provides dissipation whenever $\eta<1$, where the loss of the fraction $\eta$ comes from the force defect. We choose the weight of the corrector $\theta$ so that the momentum dissipation $\gamma\mathcal I_p(g_t)$ absorbs the positive terms, leaving a negative multiple of the entropy. We defer the proof of Proposition~\ref{prop:ent-corrector-differential} to \S~\ref{subsec:corrector-bound} and first deduce Theorem~\ref{thm:phase-space-criterion} and Corollary~\ref{cor:abstract-L2} from it.

\begin{proof}[Proof of Theorem~\ref{thm:phase-space-criterion}]
Let $g_0$ be regular and $g_t=\mathcal P_t^*g_0$. Choose $\theta=\theta_{\gamma,\eta_*}$ from \eqref{eq:ent-general-theta}. Since $\eta\leq\eta_*<1$, \eqref{eq:ent-physical-dissipation} and Proposition~\ref{prop:ent-corrector-differential} give
\begin{align*}
  -\frac\dd{\dd t}\mathcal E_\theta(g_t) &\geq (\gamma-\theta K_{\gamma,\eta_*})\mathcal I_p(g_t)+2\theta(1-\eta_*)\Ent_{\boldsymbol\mu}(g_t) -\theta\sqrt2(\gamma+\eta_*)\sqrt{\mathcal I_p(g_t)\,\Ent_{\boldsymbol\mu}(g_t)}\\
  &\geq\frac\gamma2\mathcal I_p(g_t)+2\theta(1-\eta_*)\Ent_{\boldsymbol\mu}(g_t) -\theta\sqrt2(\gamma+\eta_*)\sqrt{\mathcal I_p(g_t)\,\Ent_{\boldsymbol\mu}(g_t)}.
\end{align*}
The last line uses $\theta\leq\gamma/(2K_{\gamma,\eta_*})$ from \eqref{eq:ent-general-theta}. Young's inequality and $\theta\leq\gamma(1-\eta_*)/\bigl(2(\gamma+\eta_*)^2\bigr)$ from \eqref{eq:ent-general-theta} yield
\begin{equation*}
  \theta\sqrt2(\gamma+\eta_*)\sqrt{\mathcal I_p(g_t)\,\Ent_{\boldsymbol\mu}(g_t)} \leq\frac\gamma4\mathcal I_p(g_t) +\frac{2\theta^2(\gamma+\eta_*)^2}{\gamma}\Ent_{\boldsymbol\mu}(g_t) \leq\frac\gamma4\mathcal I_p(g_t)+\theta(1-\eta_*)\Ent_{\boldsymbol\mu}(g_t).
\end{equation*}
It follows that
\begin{equation}
  -\frac\dd{\dd t}\mathcal E_\theta(g_t) \geq\frac\gamma4\mathcal I_p(g_t)+\theta(1-\eta_*)\Ent_{\boldsymbol\mu}(g_t) \geq c_{\gamma,\eta_*}\mathcal E_\theta(g_t).
  \label{eq:ent-Lyapunov-decay}
\end{equation}
The last inequality uses $\mathcal E_\theta(g_t)\leq(1+\theta)\Ent_{\boldsymbol\mu}(g_t)$ from \eqref{eq:ent-Lyapunov-equivalence} and the definition of $c_{\gamma,\eta_*}$ in \eqref{eq:ent-general-constants}.
Then, by Gronwall's inequality and \eqref{eq:ent-Lyapunov-equivalence}, we have
\begin{equation*}
  \Ent_{\boldsymbol\mu}(g_t)\leq\frac{\mathcal E_\theta(g_t)}{1-\theta} \leq\frac{e^{-c_{\gamma,\eta_*}t}\,\mathcal E_\theta(g_0)}{1-\theta} \leq C_{\gamma,\eta_*}e^{-c_{\gamma,\eta_*}t}\Ent_{\boldsymbol\mu}(g_0),
\end{equation*}
which is \eqref{eq:ent-abstract-decay} for regular densities. The approximation argument in Appendix~\ref{app:regularity} extends the estimate to every finite-entropy initial density.
\end{proof}

\begin{proof}[Proof of Corollary~\ref{cor:abstract-L2}]
Let $\psi$ be bounded and centered. For small $|\kappa|$, Theorem~\ref{thm:phase-space-criterion} applied to $1+\kappa\psi$, with $\mathcal P_t^*(1+\kappa\psi)=1+\kappa\mathcal P_t^*\psi$, gives
\begin{equation*}
  \Ent_{\boldsymbol\mu}(1+\kappa\mathcal P_t^*\psi)\leq C_{\gamma,\eta_*}e^{-c_{\gamma,\eta_*}t}\Ent_{\boldsymbol\mu}(1+\kappa\psi).
\end{equation*}
Since $\norm{\mathcal P_t^*\psi}_\infty\le\norm\psi_\infty$, the expansion $\Ent_{\boldsymbol\mu}(1+\kappa\psi)=\frac{\kappa^2}{2}\norm\psi_2^2+o(\kappa^2)$ holds uniformly on both sides; dividing by $\kappa^2/2$ and letting $\kappa\to0$ gives
\begin{equation*}
  \norm{\mathcal P_t^*\psi}_2^2\le C_{\gamma,\eta_*} e^{-c_{\gamma,\eta_*}t}\norm\psi_2^2.
\end{equation*}
The standard density argument extends this estimate to centered $\psi\in L^2(\boldsymbol\mu)$, and then by duality, the same result holds for $\mathcal P_t$.
\end{proof}

\subsection{Bound on the kinetic derivative of the corrector}
\label{subsec:corrector-bound}
To prove Proposition~\ref{prop:ent-corrector-differential}, we integrate the kinetic derivative of the cross term over scales and bound the resulting terms by energy estimates along the Polchinski flow, applied with $f=g_t$. For this, we need the following auxiliary results: the exact evolution of the cross term at a fixed scale (Lemma~\ref{lem:ent-exact-cross-evolution}), proved in \S~\ref{subsec:corrector-derivative}, and the momentum and force-defect energy estimates for a fixed density (Lemma~\ref{lem:ent-velocity-fisher} and Proposition~\ref{prop:ent-force-energy}), proved in \S~\ref{subsec:scale-energy-estimates}.

\begin{lemma}[Exact evolution of the cross term]
\label{lem:ent-exact-cross-evolution}
Let $g_t=\mathcal P_t^*g_0$ with $g_0$ a regular density. At each fixed scale $s\geq0$, with $h=\log g_{t,s}$,
\begin{align}
  \partial_t\mathcal J_s(g_{t,s})
  ={}&\underbrace{-\mathcal I^y_s(g_{t,s})+\mathcal I^{w,A}_s(g_{t,s})}_{\text{free Hamiltonian contribution}}
     +\underbrace{\int g_{t,s}(\nabla_wh)^\top Q_s\nabla^2V_s(Q_sy)Q_s\nabla_wh\dd\boldsymbol\mu_s}_{\text{smoothed interaction Hessian}}\notag\\
  &\underbrace{-\gamma\mathcal J_s(g_{t,s})}_{\text{friction, cross term}}
     \;\underbrace{-2\gamma\int g_{t,s}\sum_{i,j}\partial_{y_i}\partial_{w_j}h\,\partial_{w_i}\partial_{w_j}h\dd\boldsymbol\mu_s}_{\text{friction, mixed Hessians}}\notag\\
  &\underbrace{-\int g_{t,s}\,r_{t,s}\cdot\nabla_yh\dd\boldsymbol\mu_s}_{\text{force defect, position score}}
     \;\underbrace{-\int g_{t,s}\sum_{i,j}\bigl[(\partial_{y_i}r_{t,s,j})\,\partial_{w_i}\partial_{w_j}h+(\partial_{w_i}r_{t,s,j})\,\partial_{y_i}\partial_{w_j}h\bigr]\dd\boldsymbol\mu_s}_{\text{force defect, derivatives}}.
  \label{eq:ent-exact-cross-evolution}
\end{align}
\end{lemma}

The \emph{free Hamiltonian contribution} in \eqref{eq:ent-exact-cross-evolution} comes from the part of the transport $-\mathcal A_s$ generated by the free Hamiltonian $\frac12\ip y{Ay}+\frac12|w|^2$, the quadratic part of the Hamiltonian, and the \emph{smoothed interaction Hessian} term from the part generated by the interaction $V_s(Q_sy)$. The latter has the form of $\mathcal I^{w,A}_s(g_{t,s})$ with $A$ replaced by $Q_s\nabla^2V_s(Q_sy)Q_s$; together with $\mathcal I^{w,A}_s(g_{t,s})$, it weights the momentum score by $A+Q_s\nabla^2V_s(Q_sy)Q_s$, the Hessian of $\frac12\ip y{Ay}+V_s(Q_sy)$. For $V_s=0$, the free Hamiltonian transport rotates the scores $(A^{-1/2}\nabla_y\log g_{t,s},\nabla_w\log g_{t,s})$ at angular frequency $\lambda^{1/2}$ on each eigenmode of $A$ with eigenvalue $\lambda$, and the cross term changes at the rate $\mathcal I^{w,A}_s(g_{t,s})-\mathcal I^y_s(g_{t,s})$, as the product $ab$ does under the rotation $\dot a=b$, $\dot b=-a$. After integration over scales, the entropy decomposition \eqref{eq:ent-decomposition} turns $-\mathcal I^y_s(g_{t,s})$ into the negative entropy $-2\Ent_{\boldsymbol\mu}(g_t)$, up to the integral of $\mathcal I^{w,A}_s(g_{t,s})$, which friction controls.

The other terms come from friction (the terms carrying $\gamma$) and from the force defect. In the coarse kinetic equation \eqref{eq:ent-coarse-density-equation}, both enter as momentum divergences $-\nabla_w^*(g_{t,s}r)$, with $r=\gamma\nabla_w\log g_{t,s}$ and $r=r_{t,s}$, respectively. Intuitively, for constant $r$, such a term translates the law $g_{t,s}\boldsymbol\mu_s$ in momentum with velocity $-r$; since the momentum marginal of $\boldsymbol\mu_s$ is $\mathsf N(0,\Id)$, the momentum score also shifts, $\partial_t\nabla_w\log g_{t,s}=(r\cdot\nabla_w)\nabla_w\log g_{t,s}-r$. This shift is the origin of the pairings $-\gamma\mathcal J_s(g_{t,s})$ and $-\int g_{t,s}\,r_{t,s}\cdot\nabla_y\log g_{t,s}\dd\boldsymbol\mu_s$ in \eqref{eq:ent-exact-cross-evolution}.

Recall from \eqref{eq:ent-physical-dissipation} that friction dissipates entropy only through $\mathcal I_p(g_t)$, the momentum Fisher information at scale zero, whereas $\frac\dd{\dd t}\mathcal J(g_t)$ involves the momentum score at every scale. Lemma~\ref{lem:ent-velocity-fisher} bounds the momentum score with its derivatives, square integrated over all scales, by $\mathcal I_p(f)$, and Proposition~\ref{prop:ent-force-energy} bounds the normalized force defect with its derivatives in the same way by $2\eta^2\Ent_{\boldsymbol\mu}(f)$. For a regular density $f$ and $f_s=\mathsf P_sf$, both estimates involve the following dissipation functional, for a smooth vector function $\psi=(\psi_1,\dots,\psi_d)$ on the coarse space,
\begin{equation}
  \mathcal D_s(\psi)=2\sum_i\int f_s\Gamma_s(\psi_i)\dd\boldsymbol\mu_s =\int f_s\left[ \sum_i\abs{\partial_{y_i}\psi}^2 +\sum_{i,j}A_{ij} \ip{\partial_{w_i}\psi}{\partial_{w_j}\psi} \right]\dd\boldsymbol\mu_s.
  \label{eq:ent-relative-dissipation}
\end{equation}

\begin{lemma}[Dissipation of momentum Fisher information]
\label{lem:ent-velocity-fisher}
For a regular density $f$, the momentum Fisher information of $f_s=\mathsf P_sf$ satisfies the dissipation identity
\begin{equation}
  \frac\dd{\dd s}\mathcal I^w_s=-\mathcal D_s(\nabla_w\log f_s)-\mathcal I^{w,A}_s.
  \label{eq:ent-velocity-fisher-identity}
\end{equation}
Consequently,
\begin{equation}
  \mathcal I^w_s\leq e^{-s}\mathcal I_p(f),\qquad s\geq0,
  \label{eq:ent-velocity-fisher-decay}
\end{equation}
and
\begin{equation}
  \int_0^\infty\bigl(\mathcal D_s(\nabla_w\log f_s)+\mathcal I^{w,A}_s\bigr)\dd s\leq \mathcal I_p(f).
  \label{eq:ent-velocity-fisher-dissipation}
\end{equation}
Moreover, for every $s\geq0$,
\begin{equation}
  \mathcal D_s(\nabla_w\log f_s)\geq\int f_s\bigl[\norm{\nabla_y\nabla_w\log f_s}_{\HS}^2+\norm{\nabla_w^2\log f_s}_{\HS}^2\bigr]\dd\boldsymbol\mu_s,
  \label{eq:ent-dissipation-Hessian}
\end{equation}
and hence
\begin{equation}
  \int_0^\infty\!\int f_s \left[ \norm{\nabla_y\nabla_w\log f_s}_{\HS}^2 +\norm{\nabla_w^2\log f_s}_{\HS}^2 \right]\dd\boldsymbol\mu_s\dd s \leq \mathcal I_p(f).
  \label{eq:ent-score-Hessian-bound}
\end{equation}
\end{lemma}

The force defect satisfies an energy estimate of the same form, with the smoothed interaction Hessian as a source.

\begin{proposition}[Energy estimate for the force defect]
\label{prop:ent-force-energy}
Let $f$ be a regular density. The force defect $R_s$ of \eqref{eq:ent-force-defect} satisfies
\begin{equation}
  \partial_sR_s =\mathsf L_sR_s-\frac12AR_s+Q_s\nabla^2V_s(Q_sy)Q_s\nabla_yf_s, \qquad R_0=0.
  \label{eq:ent-force-evolution}
\end{equation}
Moreover, the normalized force defect $r_s=R_s/f_s$ satisfies
\begin{equation}
  \int_0^\infty \left[ \mathcal D_s(r_s) +\int f_s r_s^\top Ar_s\dd\boldsymbol\mu_s \right]\dd s \leq2\eta^2\Ent_{\boldsymbol\mu}(f).
  \label{eq:ent-force-energy}
\end{equation}
In particular,
\begin{equation}
  \int_0^\infty\!\int f_s \left[ \norm{\nabla_yr_s}_{\HS}^2 +\norm{\nabla_wr_s}_{\HS}^2+\abs{r_s}^2 \right]\dd\boldsymbol\mu_s\dd s \leq2\eta^2\Ent_{\boldsymbol\mu}(f).
  \label{eq:ent-normalized-force-bound}
\end{equation}
\end{proposition}

\begin{proof}[Proof of Proposition~\ref{prop:ent-corrector-differential}]
Fix $t$ and write $f=g_t$, $f_s=g_{t,s}$, $r_s=r_{t,s}$, and $\mathcal I=\mathcal I_p(g_t)$. By \eqref{eq:ent-corrector-differentiation}, $\frac\dd{\dd t}\mathcal J(g_t)=\int_0^\infty\partial_t\mathcal J_s(g_{t,s})\dd s$, and we integrate the terms of \eqref{eq:ent-exact-cross-evolution} in $s$ one at a time.

For the free Hamiltonian contribution, by \eqref{eq:ent-decomposition} we have
\begin{equation*}
  \int_0^\infty(-\mathcal I^y_s+\mathcal I^{w,A}_s)\dd s =-2\Ent_{\boldsymbol\mu}(g_t)+2\int_0^\infty \mathcal I^{w,A}_s\dd s.
\end{equation*}
For the force-defect term with the position score, the Cauchy--Schwarz inequality, \eqref{eq:ent-normalized-force-bound}, and \eqref{eq:ent-decomposition} give
\begin{equation*}
  \left| \int_0^\infty\!\int f_sr_s\cdot\nabla_y\log f_s\dd\boldsymbol\mu_s\dd s \right| \leq\sqrt{2\eta^2\Ent_{\boldsymbol\mu}(g_t)}\sqrt{2\Ent_{\boldsymbol\mu}(g_t)}=2\eta \Ent_{\boldsymbol\mu}(g_t).
\end{equation*}
For the smoothed interaction Hessian, the definition of $b(s)$, \eqref{eq:ent-velocity-fisher-decay}, the Cauchy--Schwarz inequality in $s$, and the definition \eqref{eq:ent-abstract-budget} of $\eta$ give
\begin{align*}
  &\left| \int_0^\infty\!\int f_s(\nabla_w\log f_s)^\top Q_s\nabla^2V_s(Q_sy)Q_s\nabla_w\log f_s\dd\boldsymbol\mu_s\dd s \right|\\
  &\qquad\leq\int_0^\infty b(s)\mathcal I^w_s\dd s \leq \mathcal I\int_0^\infty b(s)e^{-s}\dd s \leq\frac{\eta}{\sqrt2}\mathcal I.
\end{align*}
For the mixed Hessians, by the Cauchy--Schwarz inequality in the indices $i,j$, Young's inequality, and \eqref{eq:ent-dissipation-Hessian}, we obtain
\begin{equation} \label{auxeq3}
  2\left| \int f_s \sum_{i,j}(\partial_{y_i}\partial_{w_j}\log f_s)(\partial_{w_i}\partial_{w_j}\log f_s) \dd\boldsymbol\mu_s \right|\leq\mathcal D_s(\nabla_w\log f_s).
\end{equation}
We estimate the friction term containing the mixed Hessians together with the momentum term of the free Hamiltonian contribution. By \eqref{auxeq3} and \eqref{eq:ent-velocity-fisher-dissipation},
\begin{align*}
  &2\int_0^\infty\mathcal I^{w,A}_s\dd s-2\gamma\int_0^\infty\!\int f_s\sum_{i,j}(\partial_{y_i}\partial_{w_j}\log f_s)(\partial_{w_i}\partial_{w_j}\log f_s)\dd\boldsymbol\mu_s\dd s\\
  \leq & 2\int_0^\infty\mathcal I^{w,A}_s\dd s+\gamma\int_0^\infty\mathcal D_s(\nabla_w\log f_s)\dd s\\
 \leq & \max\{2,\gamma\}\int_0^\infty\bigl(\mathcal D_s(\nabla_w\log f_s)+\mathcal I^{w,A}_s\bigr)\dd s\leq\max\{2,\gamma\}\,\mathcal I.
\end{align*}
For the force-defect derivative terms, the Cauchy--Schwarz inequality, \eqref{eq:ent-normalized-force-bound}, and \eqref{eq:ent-score-Hessian-bound} bound the absolute value of their scale integral by
\begin{align*}
  &\left(\int_0^\infty\!\int f_s [\norm{\nabla_yr_s}_{\HS}^2+\norm{\nabla_wr_s}_{\HS}^2] \dd\boldsymbol\mu_s\dd s\right)^{1/2}\\
  &\quad\times \left(\int_0^\infty\!\int f_s [\norm{\nabla_w^2\log f_s}_{\HS}^2+\norm{\nabla_y\nabla_w\log f_s}_{\HS}^2] \dd\boldsymbol\mu_s\dd s\right)^{1/2} \leq\sqrt2\eta\sqrt{\mathcal I\,\Ent_{\boldsymbol\mu}(g_t)}.
\end{align*}
For the friction cross term, by \eqref{eq:ent-decomposition}, $\mathcal I^w_s\leq\mathcal I^{w,A}_s$, and \eqref{eq:ent-velocity-fisher-dissipation}, we have
\begin{align*}
  \abs{\mathcal J(g_t)} &\leq \Bigl(\int_0^\infty \mathcal I^y_s\dd s\Bigr)^{1/2} \Bigl(\int_0^\infty \mathcal I^w_s\dd s\Bigr)^{1/2}\\
  &\leq\sqrt{2\Ent_{\boldsymbol\mu}(g_t)}\,\Bigl(\int_0^\infty \mathcal I^{w,A}_s\dd s\Bigr)^{1/2} \leq\sqrt{2\mathcal I\,\Ent_{\boldsymbol\mu}(g_t)},
\end{align*}
implying $-\gamma\mathcal J(g_t)\leq\sqrt2\gamma\sqrt{\mathcal I\,\Ent_{\boldsymbol\mu}(g_t)}$. Inserting these bounds into the scale integral of \eqref{eq:ent-exact-cross-evolution} gives
\begin{align*}
  \frac\dd{\dd t}\mathcal J(g_t) \leq{}& -2\Ent_{\boldsymbol\mu}(g_t)+2\eta\Ent_{\boldsymbol\mu}(g_t) +\max\{2,\gamma\}\,\mathcal I+\frac{\eta}{\sqrt2}\mathcal I\\
  &+\sqrt2\eta\sqrt{\mathcal I\,\Ent_{\boldsymbol\mu}(g_t)} +\sqrt2\gamma\sqrt{\mathcal I\,\Ent_{\boldsymbol\mu}(g_t)},
\end{align*}
which is \eqref{eq:ent-corrector-differential}. The proof is complete.
\end{proof}

\subsection{Kinetic derivative of the cross term}
\label{subsec:corrector-derivative}
We prove Lemma~\ref{lem:ent-exact-cross-evolution}, which computes the kinetic derivative of the cross term at a fixed scale.

\begin{proof}[Proof of Lemma~\ref{lem:ent-exact-cross-evolution}]
We fix $s\geq0$ and let $h=\log g_{t,s}$. Note that all integrals are against $\boldsymbol\mu_s$, which does not depend on $t$. Moreover, Lemma~\ref{lem:ent-kinetic-regularity} justifies the differentiation under the integral and the integrations by parts below.

We first differentiate the integrand $\nabla_yg_{t,s}\cdot\nabla_wg_{t,s}/g_{t,s}$ of $\mathcal J_s(g_{t,s})$ and find
\begin{equation}
  \partial_t\mathcal J_s(g_{t,s})=\int \nabla_y\partial_tg_{t,s}\cdot\nabla_wh+\nabla_w\partial_tg_{t,s}\cdot\nabla_yh-\partial_tg_{t,s}\,\nabla_yh\cdot\nabla_wh  \, \dd\boldsymbol\mu_s.
  \label{eq:ent-cross-time-derivative}
\end{equation}
The right-hand side is linear in $\partial_tg_{t,s}$, and we split it into three terms as in \eqref{eq:ent-coarse-density-equation}. As noted before Lemma~\ref{lem:ent-intertwining}, $\mathcal S_w=-\nabla_w^*\nabla_w$ by \eqref{eq:divergenceform} and $R_{t,s}=g_{t,s}r_{t,s}$ make the friction and force-defect terms momentum divergences of the same form,
\begin{equation*}
  \gamma\mathcal S_wg_{t,s}=-\nabla_w^*\bigl(g_{t,s}\,\gamma\nabla_wh\bigr),\qquad
  -\nabla_w^*R_{t,s}=-\nabla_w^*\bigl(g_{t,s}\,r_{t,s}\bigr),
\end{equation*}
so we next treat the Hamiltonian term $-\mathcal A_sg_{t,s}$ and the momentum-divergence part separately.

\smallskip

\noindent \emph{Hamiltonian transport term.} We substitute $\partial_tg_{t,s}=-\mathcal A_sg_{t,s}$ in \eqref{eq:ent-cross-time-derivative}. Since $\mathcal A_s$ is a first-order operator without zeroth-order term, we have the commutator relations
\begin{equation*}
  [\partial_{y_i},\mathcal A_s] =-\sum_j\bigl(A+Q_s\nabla^2V_s(Q_sy)Q_s\bigr)_{ij}\partial_{w_j}, \qquad [\partial_{w_i},\mathcal A_s]=\partial_{y_i},
\end{equation*}
which, along with $\nabla g_{t,s}=g_{t,s}\nabla h$ and the product rule $\mathcal A_s(\psi\chi)=\psi\,\mathcal A_s\chi+\chi\,\mathcal A_s\psi$, gives
\begin{align*}
  &\int\bigl(\nabla_y\partial_tg_{t,s}\cdot\nabla_wh+\nabla_w\partial_tg_{t,s}\cdot\nabla_yh-\partial_tg_{t,s}\,\nabla_yh\cdot\nabla_wh\bigr)\\
  &\qquad=-\int\mathcal A_s\bigl(g_{t,s}\,\nabla_yh\cdot\nabla_wh\bigr)
   +\int g_{t,s}\bigl[(\nabla_wh)^\top\bigl(A+Q_s\nabla^2V_s(Q_sy)Q_s\bigr)\nabla_wh-|\nabla_yh|^2\bigr].
\end{align*}
Since $\mathcal A_s^*=-\mathcal A_s$ in $L^2(\boldsymbol\mu_s)$ and $\mathcal A_s1=0$, the first integral vanishes. By \eqref{eq:ent-scale-scores}, the Hamiltonian transport term therefore gives the free Hamiltonian contribution and the smoothed interaction Hessian term in \eqref{eq:ent-exact-cross-evolution}:
\begin{equation*}
  -\mathcal I^y_s(g_{t,s})+\mathcal I^{w,A}_s(g_{t,s})+\int g_{t,s}(\nabla_wh)^\top Q_s\nabla^2V_s(Q_sy)Q_s\nabla_wh.
\end{equation*}

\smallskip

\noindent \emph{Momentum-divergence term.} We consider $\partial_tg_{t,s}=-\nabla_w^*(g_{t,s}r)$, where $r=\gamma\nabla_wh$ or $r=r_{t,s}$. Since $\nabla_w^*=-\diver_w+w\cdot$, we have
\begin{equation*}
  \partial_{y_i}\nabla_w^*(g_{t,s}r)=\nabla_w^*\bigl(\partial_{y_i}(g_{t,s}r)\bigr), \qquad \partial_{w_i}\nabla_w^*(g_{t,s}r)=\nabla_w^*\bigl(\partial_{w_i}(g_{t,s}r)\bigr)+g_{t,s}r_i.
\end{equation*}
Inserting these identities in \eqref{eq:ent-cross-time-derivative} and integrating by parts in $w$ gives
\begin{equation*}
  -\int\sum_{i,j}\bigl[\partial_{y_i}(g_{t,s}r_j)\,\partial_{w_i}\partial_{w_j}h+\partial_{w_i}(g_{t,s}r_j)\,\partial_{y_i}\partial_{w_j}h\bigr]
  -\int g_{t,s}\,r\cdot\nabla_yh+\int g_{t,s}\,r\cdot\nabla_w(\nabla_yh\cdot\nabla_wh).
\end{equation*}
The extra term $g_{t,s}r_i$ produces the pairing $-\int g_{t,s}\,r\cdot\nabla_yh$.
In the first integral, $\partial(g_{t,s}r_j)=g_{t,s}(\partial r_j+r_j\,\partial h)$ for $\partial=\partial_{y_i},\partial_{w_i}$, and the terms containing $r_j\,\partial h$ cancel the last integral, by
\begin{equation*}
  r\cdot\nabla_w(\nabla_yh\cdot\nabla_wh)=\sum_{i,j}r_j\bigl[(\partial_{y_i}h)\,\partial_{w_i}\partial_{w_j}h+(\partial_{w_i}h)\,\partial_{y_i}\partial_{w_j}h\bigr].
\end{equation*}
The momentum divergence $-\nabla_w^*(g_{t,s}r)$ therefore contributes
\begin{equation*}
  -\int g_{t,s}\Bigl[r\cdot\nabla_yh+\sum_{i,j}(\partial_{y_i}r_j)\,\partial_{w_i}\partial_{w_j}h+\sum_{i,j}(\partial_{w_i}r_j)\,\partial_{y_i}\partial_{w_j}h\Bigr].
\end{equation*}
For $r=r_{t,s}$, these are the two force-defect terms in \eqref{eq:ent-exact-cross-evolution}. For $r=\gamma\nabla_wh$, the first pairing contributes $-\gamma\mathcal J_s(g_{t,s})$, and the two derivative pairings coincide and together give
\begin{equation*}
  -2\gamma\int g_{t,s}\sum_{i,j}\partial_{y_i}\partial_{w_j}h\,\partial_{w_i}\partial_{w_j}h,
\end{equation*}
which are the two friction terms in \eqref{eq:ent-exact-cross-evolution}.

\smallskip

\noindent
Combining the above computations completes the proof of \eqref{eq:ent-exact-cross-evolution}.
\end{proof}

\subsection{Energy estimates along the Polchinski flow}
\label{subsec:scale-energy-estimates}

We prove Lemma~\ref{lem:ent-velocity-fisher} and Proposition~\ref{prop:ent-force-energy}. Throughout this subsection, $f$ is a fixed regular density, $f_s=\mathsf P_sf$, and the Fisher informations are given in \eqref{eq:ent-scale-scores}.

Both proofs rely on the same computational route. Specifically, for the proof of Lemma~\ref{lem:ent-velocity-fisher}, we note
that in the momentum variable, the Polchinski generator is an Ornstein--Uhlenbeck generator independent of the interaction, so that $\partial_s\nabla_wf_s=(\mathsf L_s-\frac12A)\nabla_wf_s$; see \eqref{eq:ent-momentum-gradient-evolution} below. Consequently, the momentum Fisher information decreases from $\mathcal I_p(f)$, and its rate of decrease bounds the derivatives $\nabla_y\nabla_w\log f_s$ and $\nabla_w^2\log f_s$ of the momentum score.

For the proof of Proposition~\ref{prop:ent-force-energy}, \eqref{eq:ent-force-evolution} shows that the force defect $R_s$ vanishes at scale zero and satisfies the equation of $\nabla_wf_s$ with an additional source term $Q_s\nabla^2V_s(Q_sy)Q_s\nabla_yf_s$, i.e., the smoothed interaction Hessian applied to $\nabla_yf_s$. The same computation then bounds the total dissipation of $R_s$ over all scales by $\bigl(\int_0^\infty b(s)(\mathcal I^y_s)^{1/2}\dd s\bigr)^2\leq2\eta^2\Ent_{\boldsymbol\mu}(f)$.

\begin{proof}[Proof of Lemma~\ref{lem:ent-velocity-fisher}]
Let $\psi_s=|\nabla_wf_s|^2/f_s$, and then $\mathcal I^w_s=\int\psi_s\dd\boldsymbol\mu_s$.
By Lemma~\ref{lem:ent-fixed-cutoff-regularity}, $\psi_s$ satisfies the growth condition of Lemma~\ref{lem:ent-scale-generator}, and \eqref{eq:ent-dual-scale} gives $\frac\dd{\dd s}\mathcal I^w_s=\int(\partial_s-\mathsf L_s)\psi_s\dd\boldsymbol\mu_s$. Noting that the only term of $\mathsf L_s$ in \eqref
{eq:ent-scale-generator} with coefficient that depends on $w$ is the drift $-\frac12Aw\cdot\nabla_w$, we have $[\nabla_w,\mathsf L_s]=-\frac12A\nabla_w$. Then, applying $\nabla_w$ to $\partial_sf_s=\mathsf L_sf_s$ gives
\begin{equation}
  \partial_s\nabla_wf_s=\mathsf L_s\nabla_wf_s-\frac12A\nabla_wf_s.
  \label{eq:ent-momentum-gradient-evolution}
\end{equation}
By \eqref{eq:ent-momentum-gradient-evolution} above and $\partial_sf_s=\mathsf L_sf_s$, and by the chain rule \eqref{eq:chainrule} for $\mathsf L_s$,
\begin{align*}
  \partial_s\psi_s&=2\nabla_w\log f_s\cdot\Bigl(\mathsf L_s\nabla_wf_s-\frac12A\nabla_wf_s\Bigr)-|\nabla_w\log f_s|^2\,\mathsf L_sf_s,\\
  \mathsf L_s\psi_s&=2\nabla_w\log f_s\cdot\mathsf L_s\nabla_wf_s-|\nabla_w\log f_s|^2\,\mathsf L_sf_s+2f_s\sum_i\Gamma_s(\partial_{w_i}\log f_s).
\end{align*}
The terms containing $\mathsf L_s$ cancel in the difference $ (\partial_s-\mathsf L_s)\psi_s$:
\begin{equation} \label{eq:evolutionpsis}
  (\partial_s-\mathsf L_s)\psi_s=-f_s(\nabla_w\log f_s)^\top A\,\nabla_w\log f_s-2f_s\sum_i\Gamma_s(\partial_{w_i}\log f_s).
\end{equation}
Integrating against $\boldsymbol\mu_s$ and using \eqref{eq:ent-relative-dissipation} gives \eqref{eq:ent-velocity-fisher-identity}.

Since $\mathcal I^{w,A}_s\geq\mathcal I^w_s$ for $A\succeq\Id$, it implies $\frac\dd{\dd s}\mathcal I^w_s\leq-\mathcal I^w_s$. Then, using Gronwall's inequality with $\mathcal I^w_0=\mathcal I_p(f)$ gives \eqref{eq:ent-velocity-fisher-decay}. Integrating \eqref{eq:ent-velocity-fisher-identity} over $[0,S]$ gives
\begin{equation*}
  \mathcal I^w_S+\int_0^S\bigl(\mathcal D_s(\nabla_w\log f_s)+\mathcal I^{w,A}_s\bigr)\dd s=\mathcal I_p(f),
\end{equation*}
which implies \eqref{eq:ent-velocity-fisher-dissipation} by $\mathcal I^w_S\geq0$ and letting $S\to\infty$.
Finally, \eqref{eq:ent-dissipation-Hessian} follows from \eqref{eq:ent-relative-dissipation}, by $A\succeq\Id$. Integrating it in $s$ and using \eqref{eq:ent-velocity-fisher-dissipation} with $\mathcal I^{w,A}_s\geq0$ proves \eqref{eq:ent-score-Hessian-bound}.
\end{proof}

\begin{proof}[Proof of Proposition~\ref{prop:ent-force-energy}]
We first derive the evolution \eqref{eq:ent-force-evolution} of $R_s$. By \eqref{eq:ent-gradient-defect}, $R_s=Q_s\mathsf P_s(\nabla_qf)-\nabla_yf_s$, and we will see
\begin{align}
  \partial_s\bigl(Q_s\mathsf P_s(\nabla_qf)\bigr)&=\mathsf L_s\bigl(Q_s\mathsf P_s(\nabla_qf)\bigr)-\frac12A\,Q_s\mathsf P_s(\nabla_qf), \label{auxeqa}\\
  \partial_s \nabla_y f_s & = \mathsf L_s\nabla_yf_s-\frac12A\nabla_yf_s-Q_s\nabla^2V_s(Q_sy)Q_s\nabla_yf_s. \label{auxeqb}
\end{align}
Subtracting the second equation from the first proves \eqref{eq:ent-force-evolution}.
To obtain the first equation \eqref{auxeqa}, it suffices to apply Lemma~\ref{lem:ent-scale-generator} componentwise to $\nabla_qf\in C_b^\infty(\R^{2d})$, noting that the matrix $Q_s$ does not depend on $(y,w)$, so it commutes with $\mathsf L_s$, and $\frac\dd{\dd s}Q_s=-\frac12AQ_s$. For the second one \eqref{auxeqb}, noting that the only coefficient of $\mathsf L_s$ that depends on $y$ is the drift $-\bigl(\frac12Ay+Q_s\nabla V_s(Q_sy)\bigr)\cdot\nabla_y$, we have
\begin{equation*}
  [\nabla_y,\mathsf L_s] = - \frac12A\nabla_y-Q_s\nabla^2V_s(Q_sy)Q_s\nabla_y.
\end{equation*}
Applying $\nabla_y$ to \eqref{eq:ent-scale-evolution}, with $\nabla_y\partial_sf_s=\partial_s\nabla_yf_s$ by Lemma~\ref{lem:ent-fixed-cutoff-regularity}, gives \eqref{auxeqb}. The initial value $R_0=0$ follows from \eqref{eq:ent-force-defect}, since $Q_0=\Id$, $\mathsf P_0$ is the identity, and $V_s|_{s=0}=V_0$.

Next we derive the energy estimate. For convenience, we write
\begin{equation} \label{auxdef}
  Y_s := \int\frac{\abs{R_s}^2}{f_s}\dd\boldsymbol\mu_s = \int f_s|r_s|^2\dd\boldsymbol\mu_s, \qquad \sigma_s := b(s)(\mathcal I^y_s)^{1/2}.
\end{equation}
Comparing \eqref{eq:ent-momentum-gradient-evolution} with \eqref{eq:ent-force-evolution}, a similar computation
as in the proof of Lemma~\ref{lem:ent-velocity-fisher} for \eqref{eq:evolutionpsis}, with $R_s$ in place of $\nabla_wf_s$, gives
\begin{equation*}
  (\partial_s-\mathsf L_s)\Bigl(\frac{|R_s|^2}{f_s}\Bigr) =-f_sr_s^\top Ar_s-2f_s\sum_i\Gamma_s(r_{s,i})+\underbrace{2r_s^\top Q_s\nabla^2V_s(Q_sy)Q_s\nabla_yf_s}_{\text{source}}.
\end{equation*}
By Lemma~\ref{lem:ent-fixed-cutoff-regularity}, $|R_s|^2/f_s$ satisfies the growth condition of Lemma~\ref{lem:ent-scale-generator}, and then integration against $\boldsymbol\mu_s$ implies,  by \eqref{eq:ent-dual-scale},
\begin{equation}
  \frac\dd{\dd s}Y_s =-\mathcal D_s(r_s)-\int f_s r_s^\top Ar_s\dd\boldsymbol\mu_s +2\int f_s r_s^\top Q_s\nabla^2V_s(Q_sy)Q_s\nabla_y\log f_s\dd\boldsymbol\mu_s.
  \label{eq:ent-Y-energy}
\end{equation}
Since $\norm{Q_s\nabla^2V_s(Q_sy)Q_s}_\op\leq b(s)$ for any $y$, the Cauchy--Schwarz inequality bounds the last term of \eqref{eq:ent-Y-energy}, with $\sigma_s$ given in \eqref{auxdef},
\begin{equation}
  2\int f_s r_s^\top Q_s\nabla^2V_s(Q_sy)Q_s\nabla_y\log f_s\dd\boldsymbol\mu_s\leq2b(s)\sqrt{Y_s}\,(\mathcal I^y_s)^{1/2}=2\sqrt{Y_s}\,\sigma_s.
  \label{eq:ent-source-bound}
\end{equation}
It follows that
\begin{equation}
 \frac\dd{\dd s}Y_s\leq2\sqrt{Y_s}\,\sigma_s\,, \qquad \sqrt{Y_s}\leq\int_0^s\sigma_{u}\dd u,
  \label{eq:ent-root-Y}
\end{equation}
where the second inequality follows by integrating $\frac\dd{\dd s}\sqrt{Y_s}\leq\sigma_s$ from $Y_0=0$. Now,
integrating \eqref{eq:ent-Y-energy} over $[0,S]$, using $Y_0=0$ and $Y_S\geq0$, as well as \eqref{eq:ent-source-bound}--\eqref{eq:ent-root-Y}, we obtain
\begin{align}
  \int_0^S \left[\mathcal D_s(r_s)+\int f_s r_s^\top Ar_s\dd\boldsymbol\mu_s\right]\dd s & \leq 2\int_0^S\sigma_s \left(\int_0^s\sigma_{u}\dd u\right)\dd s =\left(\int_0^S\sigma_s\dd s\right)^2 \notag \\
  & \leq \left(\int_0^Sb(s)^2\dd s\right) \left(\int_0^S\mathcal I^y_s\dd s\right) \leq2\eta^2\Ent_{\boldsymbol\mu}(f).
  \label{eq:ent-force-integrated-proof}
\end{align}
 The last line follows from the Cauchy--Schwarz inequality in $s$, the definition \eqref{eq:ent-abstract-budget} of $\eta$, and the bound $\int_0^\infty\mathcal I^y_s\dd s\leq2\Ent_{\boldsymbol\mu}(f)$ given by \eqref{eq:ent-decomposition}. Letting $S\to\infty$ gives \eqref{eq:ent-force-energy}.

 Finally, since $A\succeq\Id$, \eqref{eq:ent-relative-dissipation} gives
\begin{equation*}
  \mathcal D_s(r_s)+\int f_sr_s^\top Ar_s\dd\boldsymbol\mu_s \geq\int f_s\bigl[\norm{\nabla_yr_s}_{\HS}^2+\norm{\nabla_wr_s}_{\HS}^2+|r_s|^2\bigr]\dd\boldsymbol\mu_s,
\end{equation*}
and \eqref{eq:ent-normalized-force-bound} follows.
\end{proof}

\section{Application to the lattice sine--Gordon model}
\label{sec:SG-application}

In this section, we prove Theorem~\ref{thm:bulk-entropy} by rescaling to unit mass and bounding $\eta$ uniformly with the charge-expansion estimates of Bauerschmidt and Bodineau~\cite{BB}. We then show that the mass order of the decay rate is optimal and compare it with the static logarithmic Sobolev inequality in~\cite{BB}.

\subsection{Uniform bound on \texorpdfstring{$\eta$}{eta} and entropy decay}
\label{subsec:scaling}

\paragraph{Unit-mass coordinates.}
To apply Theorem~\ref{thm:phase-space-criterion}, we measure lengths and times in units of $m^{-1}$, rescaling the lattice by $x\mapsto mx$ and time by $t\mapsto mt$, and use Euclidean coordinates for the rescaled field and momentum:
\begin{equation}
  \tau=mt, \qquad q^{\mathrm{sc}}=m\varepsilon\varphi, \qquad p^{\mathrm{sc}}=\varepsilon\varpi.
  \label{eq:phase-scaling}
\end{equation}
The rescaling maps \eqref{eq:langevin} to the same equation at unit mass, with parameters
\begin{equation}
  (\varepsilon,L,z,m,\gamma) \longmapsto \bigl(m\varepsilon,\,mL,\,z\,m^{-(2-\alpha)},\,1,\,\gamma/m\bigr).
  \label{eq:scaling}
\end{equation}
Indeed, $\Delta_\varepsilon=m^2\Delta_{m\varepsilon}$,
\begin{equation*}
  z\varepsilon^{2-\alpha} =\bigl(z\,m^{-(2-\alpha)}\bigr)(m\varepsilon)^{2-\alpha},
\end{equation*}
and Brownian scaling turns the noise into that of the rescaled friction.
Let $\mathbf T$ be the linear phase-space map $(\varphi,\varpi)\mapsto(q^{\mathrm{sc}},p^{\mathrm{sc}})$ in \eqref{eq:phase-scaling}. It sends a density $f$ relative to $\boldsymbol\mu_{\varepsilon,L}$ to the density $f\circ\mathbf T^{-1}$ relative to $\mathbf T_\#\boldsymbol\mu_{\varepsilon,L}$, and therefore preserves the $L^2$ norm and relative entropy:
\begin{equation}
  \Ent_{\mathbf T_\#\boldsymbol\mu_{\varepsilon,L}}(f\circ\mathbf T^{-1}) =\Ent_{\boldsymbol\mu_{\varepsilon,L}}(f).
  \label{eq:scaling-entropy}
\end{equation}

As in \S~\ref{subsec:SG-main}, with a slight abuse of notation, we retain $\varepsilon$, $L$, $z$, and $\gamma$ for the rescaled parameters, set $m=1$, and rename $(\tau,q^{\mathrm{sc}},p^{\mathrm{sc}})$ as $(t,q,p)$. The dynamics become
\begin{equation}
  \dd q_t=p_t\dd t, \qquad \dd p_t=-\nabla U(q_t)\dd t-\gamma p_t\dd t +\sqrt{2\gamma}\dd W_t,
  \label{eq:unit-langevin}
\end{equation}
with standard Euclidean Brownian motion and
\begin{equation}
  U(q)=\frac12\ip q{Aq}+V_0(q), \qquad A=-\Delta_\varepsilon+\Id\succeq\Id,
  \label{eq:euclidean-energy}
\end{equation}
\begin{equation}
  V_0(q)=2z\varepsilon^{2-\alpha} \sum_x\cos\left(\frac{\sqrt\beta}{\varepsilon}q_x\right).
  \label{eq:euclidean-interaction}
\end{equation}
The invariant measure is $\boldsymbol\mu=\nu\otimes\mathsf N(0,\Id)$, where $\nu(\dd q)\propto e^{-U(q)}\dd q$. At each fixed lattice, $V_0\in C_b^\infty$, so $U$ belongs to the class \eqref{eq:potentialclass} and the regularity results of Appendix~\ref{app:regularity} apply.

\paragraph{Bound on the parameter $\eta$.}
For the unit-mass interaction \eqref{eq:euclidean-interaction}, the function $b$ and the parameter $\eta$ of \S~\ref{subsec:entropy-criterion} are
\begin{equation}
  b_{\varepsilon,L}(s) =\sup_q\norm{Q_s\nabla^2V_s(q)Q_s}_{\op}, \qquad \eta_{\varepsilon,L} =\left(\int_0^\infty b_{\varepsilon,L}(s)^2\dd s\right)^{1/2},
  \label{eq:SG-square-budget-definition}
\end{equation}
which we will estimate for the sine--Gordon model.
Let $z_{\mathrm{ce}}(\beta)>0$ be the constant $\delta_\beta$ of \cite[(3.5)]{BB}, so that at unit mass the smallness condition there reads $L\ge1$ and $|z|\le z_{\mathrm{ce}}(\beta)$. The following bound on $\eta_{\varepsilon,L}$ is uniform in the lattice spacing $\varepsilon$ and the side length $L$ for $\beta<4\pi$.

\begin{proposition}[Uniform bound on $\eta$]
\label{prop:SG-square-budget}
Let $0<\beta<4\pi$. There exists $C_\beta'>0$, depending only on $\beta$, such that, for $0<\varepsilon\leq1$, $L\geq1$, and $\abs z\leq z_{\mathrm{ce}}(\beta)$,
\begin{equation}
  \eta_{\varepsilon,L} \leq \frac{C_\beta'}{\sqrt{1-\alpha}} \abs z, \qquad \alpha=\frac{\beta}{4\pi}.
  \label{eq:SG-square-budget}
\end{equation}
Consequently, $\abs z\leq z_*(\beta)$ implies $\eta_{\varepsilon,L}\leq1/2$, where
\begin{equation}
  z_*(\beta) =\min\left\{ z_{\mathrm{ce}}(\beta), \frac{\sqrt{1-\beta/(4\pi)}}{2C_\beta'} \right\}.
  \label{eq:ent-activity-threshold}
\end{equation}
\end{proposition}

\begin{proof}
For $0<\beta<4\pi$ and $|z|\le z_{\mathrm{ce}}(\beta)$, Bauerschmidt and Bodineau bound the smoothed interaction Hessian at every scale \cite[(3.53)]{BB}. They work on the unit lattice, with the field $\varphi$ and the precision matrix $\varepsilon^2A$. Hence their scale $t$ corresponds to $s=\varepsilon^2t$, their semigroup $Q_t$ is our $Q_s$, and, since $q=\varepsilon\varphi$, their Hessian of $V_t$ is $\varepsilon^2\nabla^2V_s$. In these variables, their decay factor is $\vartheta_t^2=e^{-s}$, their length scale \cite[(3.15)]{BB} satisfies $\varepsilon\ell_t=\ell(s)$ with
\begin{equation*}
  \ell(s)=(\varepsilon\vee\sqrt s)\wedge1,
\end{equation*}
and their microscopic coupling is $O_\beta(|z|\varepsilon^{2-\alpha}\ell_t^{-\alpha})$ by \cite[(3.19)--(3.20)]{BB}. Therefore \cite[(3.53)]{BB} gives
\begin{equation*}
  b_{\varepsilon,L}(s)\leq C_\beta\abs z\,\ell(s)^{-\alpha}e^{-s},
\end{equation*}
with $C_\beta$ depending only on $\beta$. Here $\ell(s)$ is the fluctuation length at scale $s$, cut off at the lattice spacing and at the unit mass length. On the intervals $[0,\varepsilon^2]$, $[\varepsilon^2,1]$, and $[1,\infty)$, $\ell(s)$ equals $\varepsilon$, $\sqrt s$, and $1$, respectively; with $e^{-s}\leq1$ on the first two intervals, this gives the bound \eqref{eq:SG-square-profile}.
Since $0<\alpha<1$, squaring and integrating \eqref{eq:SG-square-profile} gives
\begin{align}
  \int_0^\infty b_{\varepsilon,L}(s)^2\dd s \leq C_\beta^2z^2 \left[ \varepsilon^{2-2\alpha} +\frac{1-\varepsilon^{2-2\alpha}}{1-\alpha} +\frac{e^{-2}}2 \right]
  \leq \frac{(C_\beta')^2}{1-\alpha} z^2,
  \label{eq:SG-square-integral}
\end{align}
with $(C_\beta')^2=C_\beta^2(2+e^{-2}/2)$, since $\varepsilon^{2-2\alpha}\leq1\leq(1-\alpha)^{-1}$. This proves \eqref{eq:SG-square-budget}. If $\abs z\leq z_*(\beta)$, then \eqref{eq:SG-square-budget} and \eqref{eq:ent-activity-threshold} give $\eta_{\varepsilon,L}\leq1/2$.
\end{proof}

\paragraph{Entropy decay in physical time.}
We combine the rescaling with Proposition~\ref{prop:SG-square-budget} to prove Theorem~\ref{thm:bulk-entropy}, and deduce the $L^2$ estimate.

\begin{proof}[Proof of Theorem~\ref{thm:bulk-entropy}]
At each fixed lattice, the unit-mass model \eqref{eq:unit-langevin} satisfies the assumptions of Theorem~\ref{thm:phase-space-criterion}, and Proposition~\ref{prop:SG-square-budget} gives $\eta_{\varepsilon,L}\leq1/2$ for $\abs z\leq z_*(\beta)$. Theorem~\ref{thm:phase-space-criterion} with $\eta_*=1/2$ therefore gives \eqref{eq:ent-abstract-decay} with the constants \eqref{eq:bulk-entropy-constants} in unit-mass coordinates.

The map $\mathbf T$ conjugates \eqref{eq:langevin} at time $t$ to \eqref{eq:unit-langevin} at time $\tau=mt$, with $\mathbf T_\#\boldsymbol\mu_{\varepsilon,L}=\boldsymbol\mu$. Applying \eqref{eq:scaling-entropy} at times $0$ and $t$ converts the factor $e^{-c_\gamma\tau}$ in \eqref{eq:ent-abstract-decay} to $e^{-c_\gamma mt}$. This proves \eqref{eq:bulk-entropy-decay} at $m=1$ and the stated physical-time bound for general $m$.
\end{proof}

\begin{corollary}[$L^2$ decay for lattice sine--Gordon]
\label{cor:SG-L2}
Under the assumptions of Theorem~\ref{thm:bulk-entropy}, every $\psi\in L^2(\boldsymbol\mu_{\varepsilon,L})$ with $\E_{\boldsymbol\mu_{\varepsilon,L}}\psi=0$ satisfies
\begin{equation}
  \norm{\mathcal P_t^{\varepsilon,L}\psi}_2 \le\sqrt{C_\gamma}\,e^{-c_\gamma t/2}\norm\psi_2, \qquad t\ge0.
  \label{eq:SG-L2-corollary}
\end{equation}
\end{corollary}

\begin{proof}
Apply Corollary~\ref{cor:abstract-L2} in unit-mass coordinates and use that $\mathbf T$ preserves the $L^2$ norm.
\end{proof}

\subsection{Optimality of the mass scaling}
\label{subsec:mass-calibration}
We keep $m$ explicit and use the physical parameters $\varepsilon$, $L$, $z$, and $\gamma$ of \eqref{eq:langevin}. When $\gamma$ is proportional to $m$, the constant mode of the free model evolves on the time scale $m^{-1}$. Translations of its Gaussian equilibrium law show that the order $m$ of the physical-time rate is optimal for a decay estimate uniform over the family of Theorem~\ref{thm:bulk-entropy}, which includes $z=0$. The same mode shows that, for large friction, the rate of Theorem~\ref{thm:bulk-entropy} has the optimal order in the mass and the friction. We then compare the kinetic rate with the static Poincar\'e and logarithmic Sobolev constants of the position law and with the rate of the overdamped dynamics.

Let $u$ be the constant field with $\norm u_\varepsilon=1$, so that $A_\varepsilon u=m^2u$. For $z=0$, the coordinates $(m\varphi^u,\varpi^u)$ of the constant mode, with $\varphi^{u}=\ip\varphi{u}_\varepsilon$ and $\varpi^{u}=\ip\varpi{u}_\varepsilon$, form an Ornstein--Uhlenbeck process with invariant law $\mathsf N(0,\Id_2)$ and drift matrix
\begin{equation}
  \mathbf M_{m,\gamma}=\begin{pmatrix}0&m\\-m&-\gamma\end{pmatrix},\qquad \lambda_\pm(\mathbf M_{m,\gamma})=\frac{-\gamma\pm\sqrt{\gamma^2-4m^2}}2.
  \label{eq:free-drift-matrix}
\end{equation}
Translating the equilibrium law by $h\in\R^2$ in these coordinates, and leaving every other mode at equilibrium, gives at time $t$ the equilibrium law translated by $e^{t\mathbf M_{m,\gamma}}h$, with entropy $|e^{t\mathbf M_{m,\gamma}}h|^2/2$. An estimate $\Ent(\boldsymbol\varrho_t\mid\boldsymbol\mu_{\varepsilon,L})\le Ce^{-\lambda t}\Ent(\boldsymbol\varrho_0\mid\boldsymbol\mu_{\varepsilon,L})$ valid for all these laws therefore requires $\norm{e^{t\mathbf M_{m,\gamma}}}_\op^2\le Ce^{-\lambda t}$ for all $t\geq0$, and hence
\begin{equation}
\lambda\le-2\operatorname{Re}\lambda_+(\mathbf M_{m,\gamma})=\gamma-\operatorname{Re}\sqrt{\gamma^2-4m^2}.
  \label{eq:free-rate-obstruction}
\end{equation}
When $\gamma$ is proportional to $m$, the right-hand side is proportional to $m$.

For large friction, \eqref{eq:bulk-entropy-constants} gives $c_\gamma\sim1/(8\gamma)$. Since $c_\gamma$ is evaluated at the rescaled friction, the physical-time rate $c_\gamma m$ is asymptotically $m^2/(8\gamma)$ in the physical friction, which is of the same order $m^2/\gamma$ as the right-hand side of \eqref{eq:free-rate-obstruction}.

We compare the kinetic rate with the static functional inequalities of the position law. Let $\rho_{\mathrm{LSI}}$ and $\rho_{\mathrm{PI}}$ be the optimal constants of the position law $\nu_{\varepsilon,L}$ in the conventions
\begin{align}
  \Ent_{\nu_{\varepsilon,L}}(\psi^2) &\le \frac{2}{\rho_{\mathrm{LSI}}}\mathcal D_{\varepsilon,L}(\psi),\notag\\
  \operatorname{Var}_{\nu_{\varepsilon,L}}(\psi) &\le \frac{1}{\rho_{\mathrm{PI}}}\mathcal D_{\varepsilon,L}(\psi),
  \label{eq:static-functional-conventions}
\end{align}
where
\begin{equation}
  \mathcal D_{\varepsilon,L}(\psi) :=\int\norm{\nabla_\varepsilon \psi}_\varepsilon^2\dd\nu_{\varepsilon,L} =\varepsilon^{-2}\sum_x\int|\partial_{\varphi_x}\psi|^2 \dd\nu_{\varepsilon,L}.
  \label{eq:static-dirichlet-form}
\end{equation}
For unnormalized functions, $\Ent_\nu(\psi^2)=\int\psi^2\log(\psi^2/\E_\nu\psi^2)\dd\nu$. It is known that $\rho_{\mathrm{LSI}}\le\rho_{\mathrm{PI}}$ by linearization. For $z=0$, Gaussian functional inequalities and the constant lattice mode give $\rho_{\mathrm{LSI}}=\rho_{\mathrm{PI}}=m^2$.

For small nonzero coupling, both constants remain of order $m^2$: the lower bound below is the logarithmic Sobolev inequality of Bauerschmidt and Bodineau \cite{BB}, and the upper bound follows from a variance lower bound for the spatial average of the field. Hence, for friction proportional to $m$, the physical rate of Theorem~\ref{thm:bulk-entropy} is of the order of the square root of either static constant.

\begin{proposition}[Two-sided bounds on the static constants]
\label{prop:static-mass-order}
For each $0<\beta<4\pi$, there exists $z_{\mathrm{stat}}(\beta)>0$ such that, if $0<m\varepsilon\le1$, $mL\ge1$, and
\begin{equation}
  |z|\,m^{-(2-\alpha)}\le \min\{z_*(\beta),z_{\mathrm{stat}}(\beta)\},
  \label{eq:static-comparison-regime}
\end{equation}
then
\begin{equation}
  \frac{m^2}{2}\le\rho_{\mathrm{LSI}}\le\rho_{\mathrm{PI}} \le v_*^{-1}m^2,\qquad v_*=(1-e^{-1/4})(2-e^{1/2})>0.
  \label{eq:static-mass-order}
\end{equation}
All constants are uniform in the lattice spacing and the volume.
\end{proposition}

\begin{proof}
The Hamiltonian \eqref{eq:energy}, the Dirichlet form \eqref{eq:static-dirichlet-form}, and the normalization \eqref{eq:static-functional-conventions} are those of \cite[(1.15)--(1.17) and (1.19)]{BB}. Since $z_*(\beta)\le z_{\mathrm{ce}}(\beta)$, the regime \eqref{eq:static-comparison-regime} satisfies the smallness condition of \cite[Theorem~1.6]{BB}, which is also \cite[(3.5)]{BB}, and \cite[(1.20)]{BB} gives
\begin{equation*}
  \rho_{\mathrm{LSI}}\geq m^2-C_\beta^{\mathrm{stat}}m^{\alpha}|z| =m^2\bigl(1-C_\beta^{\mathrm{stat}}|z|m^{-(2-\alpha)}\bigr),
\end{equation*}
with $C_\beta^{\mathrm{stat}}$ depending only on $\beta$. Thus $z_{\mathrm{stat}}(\beta)=1/(2C_\beta^{\mathrm{stat}})$ gives the lower bound in \eqref{eq:static-mass-order}.

For the upper bound, we bound from below the variance of the constant-mode coordinate $\varphi^u=\ip\varphi u_\varepsilon$, which satisfies $\norm{\nabla_\varepsilon\varphi^u}_\varepsilon=\norm u_\varepsilon=1$. We use the unit-mass variables $q=m\varepsilon\varphi$ of \eqref{eq:phase-scaling}. The rescaled model has parameters $(m\varepsilon,mL,z\,m^{-(2-\alpha)})$ and Gaussian precision $A=m^{-2}A_\varepsilon=-\Delta_{m\varepsilon}+\Id$. In the rest of this argument, $\nu$, $V_s$, $Q_s$, and $C_s$ refer to this model, and $b_{m\varepsilon,mL}$ and $\eta_{m\varepsilon,mL}$ are evaluated at the rescaled coupling. Let $\hat u$ be the constant Euclidean unit vector, so that $A\hat u=\hat u$ and $\varphi^u=m^{-1}\hat u\cdot q$. Under \eqref{eq:static-comparison-regime}, Proposition~\ref{prop:SG-square-budget} gives $\eta_{m\varepsilon,mL}\le1/2$, and we will use this bound to show $\operatorname{Var}_\nu(\hat u\cdot q)\ge v_*$.

Specifically, we bound this variance through the posterior covariance at a suitable scale. For $s>0$, let $\pi_s(\dd q\mid y)$ be the posterior of \S~\ref{subsec:polchinski-flow}. Averaging the posteriors against $\nu_s$ recovers $\nu$, as noted there. By \eqref{eq:polchinski2} at $Q_sy$, the posterior is the normalized integrand of
\begin{equation*}
e^{-V_s(Q_sy)}=\int e^{-V_0(q)}\,\mathsf N(Q_sy,C_s)(\dd q).
\end{equation*}
Differentiating the logarithm of this Gaussian integral twice in its mean $Q_sy$ gives
\begin{equation*}
-\nabla^2V_s(Q_sy)=C_s^{-1}\operatorname{Cov}_{\pi_s(\cdot\mid y)}(q)\,C_s^{-1}-C_s^{-1},
\end{equation*}
that is (see also \cite[Lemma~3.12, (3.61)]{BBD2024}),
\begin{equation}
  \operatorname{Cov}_{\pi_s(\cdot\mid y)}(q) =C_s-C_s\nabla^2V_s(Q_sy)C_s.
  \label{eq:coarse-conditional-covariance}
\end{equation}
Since $A\hat u=\hat u$, we have $C_s\hat u=(1-e^{-s})\hat u$ and $Q_s\hat u=e^{-s/2}\hat u$. The definition of $b_{m\varepsilon,mL}(s)$ therefore gives, for every $y$,
\begin{equation*}
\hat u^\top\nabla^2V_s(Q_sy)\hat u
  =e^s\hat u^\top Q_s\nabla^2V_s(Q_sy)Q_s\hat u
  \le e^sb_{m\varepsilon,mL}(s).
\end{equation*}
The law of total variance and \eqref{eq:coarse-conditional-covariance} give
\begin{equation}
  \begin{aligned}
  \operatorname{Var}_\nu(\hat u\cdot q)
  &\geq\E_{\nu_s}\bigl[\hat u^\top\operatorname{Cov}_{\pi_s(\cdot\mid y)}(q)\,\hat u\bigr]\\
  &\geq(1-e^{-s})-(1-e^{-s})^2e^sb_{m\varepsilon,mL}(s)\\
  &=(1-e^{-s})\bigl[1-(e^s-1)b_{m\varepsilon,mL}(s)\bigr].
  \end{aligned}
  \label{eq:coarse-static-variance-scale}
\end{equation}
Since $\int_{1/4}^{1/2}b_{m\varepsilon,mL}(s)^2\dd s\le\eta_{m\varepsilon,mL}^2\le1/4$, there exists $s\in[1/4,1/2]$ with $b_{m\varepsilon,mL}(s)\le1$. For this choice the right-hand side of \eqref{eq:coarse-static-variance-scale} is at least $v_*$, so $\operatorname{Var}_\nu(\hat u\cdot q)\ge v_*$.

Thanks to the variance lower bound shown above, noting $\varphi^u=m^{-1}\hat u\cdot q$, we obtain $\operatorname{Var}_{\nu_{\varepsilon,L}}(\varphi^u)\ge v_*m^{-2}$. As $\mathcal D_{\varepsilon,L}(\varphi^u)=1$, the Poincar\'e inequality in \eqref{eq:static-functional-conventions} gives
\begin{equation*}
  \rho_{\mathrm{PI}}\leq\frac{\mathcal D_{\varepsilon,L}(\varphi^u)}{\operatorname{Var}_{\nu_{\varepsilon,L}}(\varphi^u)}\leq v_*^{-1}m^2.\qedhere
\end{equation*}
\end{proof}

Finally, we compare with the overdamped dynamics, whose rate depends on a choice of time normalization; we fix it by a mobility $\omega>0$. The overdamped equation with mobility $\omega$ and the same position law is
\begin{equation*}
  \dd\varphi_t=-\omega\nabla_\varepsilon H_{\varepsilon,L}(\varphi_t)\dd t +\sqrt{2\omega}\dd W_t^\varepsilon.
\end{equation*}
Its logarithmic Sobolev estimate gives entropy decay with exponent $2\omega\rho_{\mathrm{LSI}}$. In the free model this exponent is sharp: translations of equilibrium in the constant mode have entropy proportional to $e^{-2\omega m^2t}$. With unit mobility $\omega=1$, the overdamped exponent is of order $m^2$, and the kinetic rate $m$ is of the order of its square root. The mobility $\omega=1/\gamma$ is that of the large-friction limit of \eqref{eq:langevin}; with this choice, the overdamped exponent $2m^2/\gamma$ of the constant mode matches the large-friction asymptotics of the right-hand side of \eqref{eq:free-rate-obstruction} and is of order $m$ when the friction is proportional to $m$.

\appendix

\section{Regularity, scale limits, and approximation}
\label{app:regularity}

In this appendix, we justify the differentiations and scale limits in Sections~\ref{sec:scale-calculus}--\ref{sec:modified-entropy}, and extend the entropy estimate to finite-entropy initial laws by approximation. All regularity bounds are at a fixed cutoff and may depend on the model parameters.

\subsection{Regularity and the scale limit}
The kinetic stochastic flow and the Gaussian convolutions defining the coarse-grained densities provide the bounds needed for differentiation and integration by parts, and for the limit $s\to\infty$ in the entropy decomposition. Throughout, the matrix $A$, potential $V_0$, and parameter $\gamma>0$ are as in \S~\ref{subsec:kinetic-dynamics}.

\begin{lemma}[Regularity along the Polchinski flow]
\label{lem:ent-fixed-cutoff-regularity}
The renormalized potential satisfies $$|\partial_s[V_s(Q_sy)]|\leq C(1+|y|)$$ for $s\geq0$ and $y\in\R^d$. Let $f$ be a regular density with $0<c_-\leq f\leq c_+$. Then $c_-\leq f_s\leq c_+$, and for $k=0,1,\ldots$, $s\geq0$, and $(y,w)\in\R^{2d}$,
\begin{equation}
  |\nabla_{(y,w)}^kf_s(y,w)| \leq C_ke^{-ks/2},\qquad
  |\partial_s\nabla_{(y,w)}^kf_s(y,w)| \leq C_ke^{-ks/2}(1+|(y,w)|),
  \label{eq:ent-core-scale-decay}
\end{equation}
where $C_k$ depends on $f$ only through $c_-$ and the norms $\|\nabla_{(q,p)}^jf\|_\infty$. The force defect $R_s$ and the normalized force defect $r_s$ satisfy the same bounds with an additional factor $e^{-s/2}$. As a result, $f_s\log f_s$, $|\nabla_wf_s|^2/f_s$, and $|R_s|^2/f_s$ satisfy the growth condition of Lemma~\ref{lem:ent-scale-generator}.
\end{lemma}

\begin{proof}
Conditional expectation preserves $c_-\leq f_s\leq c_+$. In \eqref{eq:ent-coarse-operator}, the denominator is at least $e^{-\sup V_0}$, and each derivative in $(y,w)$ contributes a factor $Q_s$ with $\norm{Q_s}_\op\leq e^{-s/2}$. The differentiated Gaussian averages are bounded, proving the first bound in \eqref{eq:ent-core-scale-decay}.

\paragraph{Differentiation in scale.}
The Polchinski equation \eqref{eq:polchinski-equation} and the chain rule, with $\frac\dd{\dd s}Q_s=-\frac12AQ_s$, give
\begin{equation}
  \partial_s[V_s(Q_sy)] =\frac12Q_s^2:\nabla^2V_s(Q_sy) -\frac12\abs{Q_s\nabla V_s(Q_sy)}^2 -\frac12Ay\cdot Q_s\nabla V_s(Q_sy).
  \label{eq:ent-W-evolution}
\end{equation}
Together with \eqref{eq:ent-scale-evolution}, this provides the required scale derivatives. In particular, $\nabla V_s(Q_sy)$ is the posterior average \eqref{eq:posterior-average} of $\nabla V_0$, so $\norm{\nabla V_s}_\infty\leq\norm{\nabla V_0}_\infty$ and $|Q_s\nabla V_s(Q_sy)|\leq e^{-s/2}\norm{\nabla V_0}_\infty$.

The second-order coefficients of $\mathsf L_s$ are constant, and its drift grows at most linearly in $(y,w)$. Differentiating the Gaussian formula for $e^{-V_s}$ and using $e^{-V_s}\geq e^{-\sup V_0}$ bounds every derivative of $V_s$ uniformly in $s$. Thus $|\nabla_y^k[Q_s\nabla V_s(Q_sy)]|\leq C_ke^{-(k+1)s/2}$. Applying $\nabla_{(y,w)}^k$ to $\partial_sf_s=\mathsf L_sf_s$ gives the second bound in \eqref{eq:ent-core-scale-decay}; the linear factor comes from the drift. Equation \eqref{eq:ent-W-evolution} gives the stated bound on $\partial_s[V_s(Q_sy)]$.

\paragraph{Force defect.}
The force defect $R_s$ is a difference of Gaussian averages with an additional outer factor $Q_s$. The same quotient calculation, using $f_s\geq c_-$, gives
\begin{equation*}
  |\nabla_{(y,w)}^kR_s|+|\nabla_{(y,w)}^kr_s| \leq C_ke^{-(k+1)s/2},\qquad k\geq0,
\end{equation*}
and \eqref{eq:ent-force-evolution} gives the corresponding scale derivatives.

\paragraph{Differentiation under the integral.}
The scale derivative of $Z_0^{-1}e^{-V_s(Q_sy)}$ is bounded by $C(1+|y|)$ times this density. Together with $f_s\geq c_-$ and the force-defect bounds, the derivative estimates show that $f_s\log f_s$, $|\nabla_wf_s|^2/f_s$, and $|R_s|^2/f_s$, with their scale derivatives and their $(y,w)$-derivatives of order at most two, grow at most linearly in $(y,w)$, uniformly in $s\geq0$, which gives the growth condition of Lemma~\ref{lem:ent-scale-generator}. The same estimates give polynomial majorants for the other integrands in Sections~\ref{sec:scale-calculus}--\ref{sec:modified-entropy} and their scale derivatives, locally uniformly in $s$. Gaussian moments justify differentiation under the integral, and compact spatial cutoffs and dominated convergence justify integration by parts.
\end{proof}

\begin{lemma}[Regularity along the kinetic evolution]
\label{lem:ent-kinetic-regularity}
Let $g_0$ be a regular density with $0<c_-\leq g_0\leq c_+$, and let $T<\infty$. Then $g_t=\mathcal P_t^*g_0$ is a regular density for $0\leq t\leq T$, with
\begin{equation}
  c_-\leq g_t\leq c_+, \qquad \sup_{0\leq t\leq T}\norm{\nabla_{(q,p)}^kg_t}_\infty<\infty \quad(k\geq1),
  \label{eq:ent-core-derivatives}
\end{equation}
and hence Lemma~\ref{lem:ent-fixed-cutoff-regularity} applies to $f=g_t$ with constants uniform in $t\in[0,T]$. Moreover, for $k=0,1,\ldots$,
\begin{equation}
  \begin{gathered}
  \sup_{0\leq t\leq T}|\partial_t\nabla_{(q,p)}^kg_t(q,p)| \leq C_{T,k}(1+|q|+|p|),\\
  \sup_{0\leq t\leq T}|\partial_t\nabla_{(y,w)}^kg_{t,s}(y,w)| \leq C_{T,k}e^{-ks/2}(1+|(y,w)|),
  \end{gathered}
  \label{eq:ent-kinetic-time-derivatives}
\end{equation}
and $R_{t,s}$ and $r_{t,s}$ satisfy the second bound with an additional factor $e^{-s/2}$.
\end{lemma}

\begin{proof}
The reversed drift
\begin{equation*}
  (q,p)\longmapsto(-p,\nabla U(q)-\gamma p)
\end{equation*}
is smooth and globally Lipschitz, with bounded derivatives of positive order. With momentum noise $\sqrt{2\gamma}\dd W_t$, it generates $\mathcal P_t^*$, whose Markov property preserves $c_-\leq g_t\leq c_+$. The derivatives of its stochastic flow satisfy linear differential equations; induction and Gronwall bound them on $[0,T]$, uniformly in the initial point. Differentiating the expectation defining $\mathcal P_t^*g_0$ proves \eqref{eq:ent-core-derivatives}.
Since the drift grows at most linearly, the density equation, with the bounded spatial derivatives of $g_t$, gives the first bound in \eqref{eq:ent-kinetic-time-derivatives}. For the time derivative of $g_{t,s}$, replace $g_t$ by $\partial_tg_t$ in the numerator of \eqref{eq:ent-coarse-operator}. Its derivatives have at most linear growth, while $C_s\preceq A^{-1}\preceq\Id$ and $AC_s\preceq\Id$ give uniformly bounded Gaussian moments. Quotient differentiation proves the bound on $\partial_t\nabla_{(y,w)}^kg_{t,s}$, and the same calculation gives the bounds on the time derivatives of $R_{t,s}$ and $r_{t,s}$.
\end{proof}

\begin{lemma}[Differentiation in kinetic time]
\label{lem:ent-kinetic-differentiation}
Let $g_0$ be a regular density and $T<\infty$. The integral defining the corrector $\mathcal J(g_t)$ converges absolutely, $t\mapsto\mathcal J(g_t)$ belongs to $C^1([0,T])$, and
\begin{equation}
  \frac\dd{\dd t}\mathcal J(g_t) =\int_0^\infty\!\int \partial_tj_{t,s}\dd\boldsymbol\mu_s\dd s,
  \label{eq:ent-corrector-differentiation}
\end{equation}
where $j_{t,s}=g_{t,s}\nabla_y\log g_{t,s}\cdot\nabla_w\log g_{t,s}$ is the integrand of the cross term $\mathcal J_s(g_{t,s})$. Moreover, $t\mapsto\Ent_{\boldsymbol\mu}(g_t)$ belongs to $C^1([0,T])$ and satisfies \eqref{eq:ent-physical-dissipation}.
\end{lemma}

\begin{proof}
Divide the first derivatives of $g_{t,s}$ by its positive lower bound and use Lemma~\ref{lem:ent-kinetic-regularity} to obtain
\begin{equation*}
  |\nabla_{(y,w)}\log g_{t,s}|\leq C_Te^{-s/2},\qquad |\partial_t\nabla_{(y,w)}\log g_{t,s}| \leq C_Te^{-s/2}(1+|(y,w)|).
\end{equation*}
By the product rule, the integrand $j_{t,s}$ of the cross term satisfies
\begin{equation}
  \sup_{0\leq t\leq T} \bigl(|j_{t,s}(y,w)|+|\partial_tj_{t,s}(y,w)|\bigr) \leq C_Te^{-s}(1+|(y,w)|).
  \label{eq:ent-cross-domination}
\end{equation}
Moreover,
\begin{equation*}
  e^{-\operatorname{osc}V_0} \leq\frac{\dd\boldsymbol\mu_s}{\dd\mathsf N(0,\mathbf C_\infty)} \leq e^{\operatorname{osc}V_0},\qquad s\geq0.
\end{equation*}
Thus \eqref{eq:ent-cross-domination} provides a common integrable majorant on $[0,\infty)\times\R^{2d}$, relative to $\dd s\,\mathsf N(0,\mathbf C_\infty)(\dd y\,\dd w)$. Dominated differentiation proves \eqref{eq:ent-corrector-differentiation} and continuity of its right-hand side.

The two-sided bounds on $g_t$, the linear growth bound for $\partial_tg_t$ in \eqref{eq:ent-kinetic-time-derivatives}, and Gaussian moments justify differentiating $\Ent_{\boldsymbol\mu}(g_t)$. By skew-adjointness of the Hamiltonian part and Gaussian integration by parts in momentum, we obtain \eqref{eq:ent-physical-dissipation}.
\end{proof}

\begin{lemma}[Vanishing terminal entropy]
\label{lem:ent-scale-endpoints}
For a regular density $f$, the coarse-grained densities satisfy $f_s\to1$ pointwise and $\Ent_{\boldsymbol\mu_s}(f_s)\to0$ as $s\to\infty$.
\end{lemma}
\begin{proof}
In \eqref{eq:ent-coarse-operator}, $Q_s\to0$ and $\mathbf C_s\to\mathbf C_\infty$. Since $f$ and $V_0$ are bounded and continuous, the expectation converges to $\E_{\mathsf N(0,\mathbf C_\infty)}[e^{-V_0}f]=Z_0\E_{\boldsymbol\mu}f=Z_0$, while $e^{-V_s(Q_sy)}\to Z_0$. Hence $f_s\to1$. Conditional expectation gives $\inf f\leq f_s\leq\sup f$, and Gaussian convolution gives $\inf V_0\leq V_s\leq\sup V_0$. Writing the entropy with the density of $\boldsymbol\mu_s$ relative to $\mathsf N(0,\mathbf C_\infty)$,
\begin{equation*}
  \Ent_{\boldsymbol\mu_s}(f_s)
  =Z_0^{-1}\int f_s\log f_s\,e^{-V_s(Q_sy)}\,\mathsf N(0,\mathbf C_\infty)(\dd y\,\dd w),
\end{equation*}
the integrand is therefore uniformly bounded and tends to zero pointwise. Dominated convergence proves the claim.
\end{proof}

\subsection{Approximation of finite-entropy densities}
\label{subsec:regularity-approximation}

We complete the proof of Theorem~\ref{thm:phase-space-criterion}
by approximating the initial density in $L^1$ and in entropy.
The decay constants for regular densities are independent of their
derivative bounds and positive lower bounds.

\begin{proof}[Completion of the proof of Theorem~\ref{thm:phase-space-criterion}]
Let $g_0$ be a probability density relative to $\boldsymbol\mu$
with finite entropy. Truncating $g_0$ at level $n$, adding $n^{-1}$,
mollifying in $(q,p)$ at a scale depending on $n$, and normalizing
gives regular densities $g_0^{(n)}$ that converge to $g_0$ in
$L^1(\boldsymbol\mu)$ and in entropy, by dominated convergence.
Since $\mathcal P_t^*$ is a contraction on $L^1(\boldsymbol\mu)$,
$g_t^{(n)}=\mathcal P_t^*g_0^{(n)}\to g_t$ in $L^1(\boldsymbol\mu)$,
and lower semicontinuity of entropy and the decay estimate for
regular densities give
\[
  \Ent_{\boldsymbol\mu}(g_t)
  \leq\liminf_{n\to\infty}\Ent_{\boldsymbol\mu}(g_t^{(n)})
  \leq C_{\gamma,\eta_*}e^{-c_{\gamma,\eta_*}t}\Ent_{\boldsymbol\mu}(g_0).\qedhere
\]
\end{proof}

\section{Identification with the formulation of Bauerschmidt and Bodineau}
\label{app:BB-identification}
The change of variables $\mathbf Q_s$ identifies the construction in Section~\ref{sec:scale-calculus} with the formulation of Bauerschmidt and Bodineau~\cite{BB} and the survey~\cite{BBD2024}. We give the correspondence for the measure, semigroup, generator, and entropy decomposition.

The pushforward $(\mathbf Q_s)_\#\boldsymbol\mu_s=Z_0^{-1}e^{-V_s(q)}\mathsf N(0,\mathbf Q_s\mathbf C_\infty\mathbf Q_s)(\dd q\,\dd p)$ is the renormalized measure for the block decomposition \eqref{eq:ent-BB-block-covariance} in \cite[(1.5)]{BB} and \cite[Definition~3.2]{BBD2024}. The function $f_s\circ\mathbf Q_s^{-1}$ is the image of $f$ under the Polchinski semigroup \cite[(3.12)]{BBD2024} from scale $0$ to $s$. Our scale pair $(0,s)$ and covariance $\mathbf C_s$ correspond to $(0,t)$ and its covariance in that definition. In our variables $(y,w)$, the Gaussian reference measure remains $\mathsf N(0,\mathbf C_\infty)$.

Applying the chain rule to \eqref{eq:ent-scale-evolution}, with $\frac\dd{\dd s}\mathbf Q_s^{-1}=\frac12\operatorname{diag}(A,A)\mathbf Q_s^{-1}$, gives the generator of $f_s\circ\mathbf Q_s^{-1}$ in $(q,p)$:
\begin{equation*}
  \frac12\Bigl(\frac\dd{\dd s}\mathbf C_s\Bigr):\nabla_{(q,p)}^2-(Q_s^2\nabla V_s)\cdot\nabla_q.
\end{equation*}
The Ornstein--Uhlenbeck drifts cancel against the derivative of $\mathbf Q_s^{-1}$. This is the Polchinski generator \cite[(2.8)]{BB}, with $\bigl(\frac\dd{\dd s}\mathbf C_s\bigr):\nabla_{(q,p)}^2=Q_s^2:\nabla_q^2+AQ_s^2:\nabla_p^2$ the Laplacian associated with $\frac\dd{\dd s}\mathbf C_s$. Since $Q_s$ commutes with $A$, \eqref{eq:ent-BB-block-covariance} gives
\begin{equation}
  \mathbf Q_s^{-1}\Bigl(\frac\dd{\dd s}\mathbf C_s\Bigr)\mathbf Q_s^{-1}
  =\begin{pmatrix}\Id&0\\0&A\end{pmatrix}.
  \label{eq:ent-coarse-diffusion}
\end{equation}
For smooth $\psi$ on the coarse space, $\nabla_{(q,p)}(\psi\circ\mathbf Q_s^{-1})=\mathbf Q_s^{-1}(\nabla_{(y,w)}\psi)\circ\mathbf Q_s^{-1}$, so \eqref{eq:ent-coarse-diffusion} gives
\begin{equation}
  \frac12\Bigl\langle\nabla_{(q,p)}(\psi\circ\mathbf Q_s^{-1}), \Bigl(\frac\dd{\dd s}\mathbf C_s\Bigr)\nabla_{(q,p)}(\psi\circ\mathbf Q_s^{-1})\Bigr\rangle\circ \mathbf Q_s =\frac12|\nabla_y\psi|^2+\frac12(\nabla_w\psi)^\top A\,\nabla_w\psi=\Gamma_s(\psi).
  \label{eq:ent-BB-Fisher-conversion}
\end{equation}
The left-hand side is the carr\'e du champ of the generator of Bauerschmidt and Bodineau, composed with $\mathbf Q_s$. Since $\Ent_{(\mathbf Q_s)_\#\boldsymbol\mu_s}(f_s\circ\mathbf Q_s^{-1})=\Ent_{\boldsymbol\mu_s}(f_s)$, taking $\psi=\log f_s$ in \eqref{eq:ent-BB-Fisher-conversion} transforms the identity \cite[(2.16)--(2.17)]{BB},
\begin{equation*}
  \Ent_{\boldsymbol\mu}(f)=\frac12\int_0^\infty\!\!\int (f_s\circ\mathbf Q_s^{-1})\Bigl\langle\nabla_{(q,p)}\log(f_s\circ\mathbf Q_s^{-1}), \Bigl(\frac\dd{\dd s}\mathbf C_s\Bigr)\nabla_{(q,p)}\log(f_s\circ\mathbf Q_s^{-1})\Bigr\rangle \dd(\mathbf Q_s)_\#\boldsymbol\mu_s\dd s,
\end{equation*}
into \eqref{eq:ent-decomposition}.

\bibliography{references}

\end{document}